\documentclass[aps,preprintnumbers,nofootinbib,letterpaper,twocolumn]{revtex4}

\usepackage{epsfig,graphics}
\usepackage{amsfonts,amssymb,amsmath} 
\usepackage{amsthm}         
\usepackage{color}
\newcommand{\chapter}{paper}
\newcommand{\comment}[1]{}
\newcommand{\ket}[1]{\left | #1 \right\rangle}
\newcommand{\bra}[1]{\left \langle #1 \right |}

\newcommand{\ketbra}[2]{\ket{#1}\!\!\bra{#2}}
\newcommand{\proj}[1]{\ketbra{#1}{#1}}
\newcommand{\bef}{\rightsquigarrow}
\newcommand{\nbef}{\not\rightsquigarrow}

\newcommand{\id}{\openone}
\newcommand{\ot}{\otimes}
\newcommand{\tr}{\mathrm{tr}}
\newcommand{\cA}{\mathcal{A}}
\newcommand{\cB}{\mathcal{B}}

\newcommand{\cH}{\mathcal{H}}

\newcommand{\cX}{\mathcal{X}}

\newcommand{\up}{\uparrow}
\newcommand{\dn}{\downarrow}

\theoremstyle{plain}
\newtheorem{theorem}{Theorem}
\newtheorem{lemma}{Lemma}

\newtheorem{claim}{Claim}

\theoremstyle{definition}
\newtheorem{definition}{Definition}

\begin{document}
\title{The completeness of quantum theory for predicting measurement
  outcomes}

\author{Roger \surname{Colbeck}}
\email[]{colbeck@phys.ethz.ch}
\affiliation{Institute for Theoretical Physics, ETH Zurich, 8093
Zurich, Switzerland.}
\author{Renato \surname{Renner}}
\email[]{renner@phys.ethz.ch}
\affiliation{Institute for Theoretical Physics, ETH Zurich, 8093
Zurich, Switzerland.}

\date{$11^{\text{th}}$ July 2013}

\begin{abstract}
  The predictions that quantum theory makes about the outcomes of
  measurements are generally probabilistic. This has raised the
  question whether quantum theory can be considered complete, or
  whether there could exist alternative theories that provide improved
  predictions. Here we review recent work that considers arbitrary
  alternative theories, constrained only by the requirement that they
  are compatible with a notion of ``free choice'' (defined with
  respect to a natural causal order). It is shown that quantum theory is
  ``maximally informative'', i.e., there is no other compatible theory
  that gives improved predictions. Furthermore, any alternative
  maximally informative theory is necessarily equivalent to quantum
  theory. This means that the state a system has in such a theory is
  in one-to-one correspondence with its quantum-mechanical state (the
  wave function). In this sense, quantum theory is complete.
\end{abstract}

\maketitle

\section{Introduction}

In this \chapter{} we look at the question of whether quantum theory is
optimal in terms of the predictions it makes about measurement
outcomes, or whether, instead, there could exist an alternative theory
with improved predictive power.  This was much debated in the early
days of quantum theory, when many eminent physicists supported the
view that quantum theory will eventually be replaced by a deeper
underlying theory.  Our aim will be to show that no alternative theory
can extend the predictive power of quantum theory, and hence that, in
this sense, quantum theory is complete.

Before turning to this question, it is worth reflecting on why one
might think that quantum theory may not be optimally predictive.  A
key factor is that the theory is probabilistic.  This is in stark
contrast with classical theory, which is deterministic at a
fundamental level.  Even in classical theory there are scenarios where
we may assign probabilities to various events, for example when making
a weather forecast.  However, this isn't in conflict with our belief
in underlying determinism, but, instead, the fact that we assign
probabilities simply reflects a lack of knowledge (about the precise
value of certain physical quantities) when making the prediction.  By
analogy, we might imagine that even if we know the quantum state of a
system before measurement (i.e., its wave function), we are also in a
position of incomplete knowledge, and that additional knowledge might
be provided in a higher theory.

A further argument for incompleteness was given by Einstein, Podolsky
and Rosen (EPR)~\cite{EPR}.  They argued that whenever the outcome of
an experiment can be predicted with certainty, there should be a
counterpart in the theory representing its value.  They then consider
measurements on a maximally entangled pair.  In this scenario, the
outcome of any measurement on one member of the pair can be perfectly
predicted given access to the other member.  Since the particles can
be far apart, a measurement on one shouldn't, say EPR, affect the
other in any way.  They hence argue that there should be parts of the
theory allowing these perfect predictions and, hence, that the quantum
description is incomplete.

Following EPR, one might hope that quantum theory can be explained in
terms of an underlying deterministic theory.  Such a view was put into
doubt by the Bell-Kochen-Specker theorem, independently discovered by
Kochen and Specker~\cite{KS} and by Bell~\cite{Bell_KS}, who showed
that an underlying deterministic theory is not possible if one demands
non-contextuality and freedom of choice.  (A non-contextual theory is
one in which the probability of a particular measurement outcome
occurring depends only on the projector associated with that outcome,
and not on the entire set of projectors that specify the measurement
according to quantum theory.)  Furthermore it was also shown by
Bell~\cite{Bell} that there cannot be an underlying theory that is
compatible with \emph{local causality} (we will explain this in more
detail in Section~\ref{sec_constraints}).  It is also worth noting
that an assumption about locality can be seen as a physical means of
justifying certain non-contextuality conditions.

In this \chapter{}, we consider arbitrary alternative theories and ask
whether they could have more predictive power than quantum theory.  We
remark that this question is different from those asked by Kochen and
Specker and by Bell, whose goal was to rule out theories with certain
specific properties such as non-contextuality or local causality. In
this work, we do not demand any of these properties. The only
assumption we make about a theory is that it is compatible with a
notion of free choice (defined with respect to a natural
causal order|see later).  Roughly, the freedom of choice assumption
demands that the theory can be applied to a setting where an
experimenter makes certain choices independently of certain
pre-existing parameters.  It is worth noting that quantum theory is
compatible with this assumption, as we would expect, since it is a
reasonable theory.  We also remark that such an assumption is
necessary for Kochen and Specker's as well as for Bell's arguments.

%In the context of this question, the results by Kochen and
%Specker~\cite{KS} and by Bell~\cite{Bell_KS,Bell} rule out a large set
%of deterministic theories (those that have free choice and that are
%non-contextual or locally deterministic). However, they leave open the
%possibility of an alternative theory that enables improved predictions
%over those of quantum theory, but which may still be probabilistic.
As a toy example of an alternative theory that enables improved
predictions over those of quantum theory, but which may still be
probabilistic, one might imagine that the quantum state is
supplemented by an additional parameter $Z$.  When measuring one half
of a maximally entangled pair of qubits, it could be that if $Z=0$ the
extended theory assigns outcome $0$ with probability $3/4$, and
outcome $1$ with probability $1/4$, while, if $Z=1$, the extended
theory assigns outcome $0$ with probability $1/4$, and outcome $1$
with probability $3/4$.  The extended theory would thus provide more
information than quantum theory, which predicts that both outcomes
occur with probability $1/2$.  Furthermore, if $Z$ is uniformly
distributed, the quantum predictions are recovered when $Z$ is unknown
(and hence the extended theory is compatible with quantum theory).

This particular example is rather artificial and its purpose is merely
to illustrate that|in principle|a theory that is more informative than
quantum theory is conceivable.  However, there are historical
precedents of this type, for instance related to the problem of
determining the mass of chemical elements. Take, as an example, the
atomic mass of chlorine.  Before the discovery of isotopes, its atomic
mass was thought to be $35.5$, and the standard measurement techniques
of the time confirmed it as such.  However, it was later discovered
that chlorine in fact naturally occurs as two isotopes with atomic
masses $35$ and $37$ (in approximate ratio $3:1$).  By introducing
isotopes, the theory was extended in such a way that the mass of an
individual atom could be better predicted.  Note that the predictions
made before the discovery of isotopes were not incorrect, but are
simply the natural ones to make without knowledge of the different
isotopes (and hence the new theory is compatible with the old one).

Returning to quantum theory, various alternatives, motivated more
physically than our earlier toy example, have been proposed in the
past, some of which we will review later (see
Section~\ref{sec_constraints}). Similarly to quantum theory, these
alternatives provide rules to compute predictions for future
measurement outcomes, based on certain (additional) parameters.

The aim of this \chapter{} is to explain recent results relating the
predictive power of quantum theory to that of possible alternative
theories~\cite{CR_ext,CR_wavefn}.  For this, we first need to specify
what we mean by ``quantum theory'' and by ``alternative theories'',
and how they can be compared (Section~\ref{sec_theories}). The central
requirement we impose on any alternative theory is that it be
compatible with a notion of ``free choice''. This means that the
theory can be applied consistently in scenarios where measurements are
chosen independently of certain other events
(Section~\ref{sec_freedom}).  We then discuss the implications of some
existing results to our main question.  These impose constraints on
any alternative theory that is compatible with quantum theory; for
instance, no such theory can be locally deterministic
(Section~\ref{sec_constraints}).  The last sections are then devoted
to the recent, more general, results.  A central claim is that no
alternative theory that is compatible with quantum theory can improve
the predictions of quantum theory (Sections~\ref{sec_extended}
and~\ref{sec_proof}). Furthermore, if such an alternative theory is
also at least as informative as quantum theory, then it is necessarily
equivalent to quantum theory (Section~\ref{sec_equivalence}). In this
sense, quantum theory is complete. We conclude with a discussion of
how these results relate to known hidden-variable theories, in
particular the de Broglie-Bohm theory, and mention some applications
(Section~\ref{sec_discussion}).

\section{Preliminaries}

\subsection{Notation}

On a technical level, the main results presented in this \chapter{} are
theorems about random variables (RVs) whose (joint) probability
distribution satisfies certain assumptions. We will only use RVs with
discrete range.  In the following we introduce our notation for such
RVs and their distributions.

We usually use upper case letters to denote RVs, while lower case
letters specify particular values they can take.  Thus, $X=x$ means
that the RV $X$ takes the value $x$.  We write $P_X$ to denote the
probability distribution of the RV $X$, with $P_X(x)$ being the
probability that $X=x$.  For two RVs, $X$ and $Y$, $P_{XY}$ represents
their joint distribution.  We also use $P_{X|Y}:=P_{XY}/P_Y$ to
represent the conditional distribution of $X$ given $Y$. This is
defined for all $y$ such that $P_Y(y) > 0$.  For any such $y$, we
write $P_{X|Y=y} := P_{X|Y}(\cdot, y)$ to denote the distribution of
the RV $X$ conditioned on $Y=y$. We often abbreviate this distribution
to $P_{X|y}$.  We also use $P(X=Y)$ to denote the probability that the
RVs $X$ and $Y$ have equal values, i.e.\ $P(X=Y):=\sum_xP_{XY}(x,x)$
and, likewise, $P(X\neq Y):=1-P(X=Y)$.

\subsection{Distance between probability distributions}

Our technical argument uses the \emph{variational distance} to
quantify the closeness of two probability distributions. For two
distributions, $P_X$ and $Q_X$, it is defined by
\[
  D(P_X,Q_X):=\frac{1}{2}\sum_x|P_X(x)-Q_X(x)|\, .
\]
This measure is connected to the distinguishability of the two
distributions.  Specifically, suppose we have a black box that samples
either from $P_X$ or $Q_X$. Then, given one sample, the maximum
probability of successfully guessing whether the sample has been
generated from $P_X$ or $Q_X$ equals $\frac{1}{2}(1+D(P_X,Q_X))$.
Thus, if two distributions are close in variational distance, they are
virtually indistinguishable.  Appendix~\ref{sec_dist} summarizes some
properties of $D(\cdot, \cdot)$ that are used in this work.

\subsection{Measuring correlations}\label{sec_CB}

A useful approach towards characterizing alternative theories is to
consider the correlations (between the outcomes of two distant
measurements) that can be reproduced by a given theory. The strength
of these correlations may then, for instance, be compared to those
occurring in quantum theory. To quantify correlations, we use a
measure that has been proposed by Pearle~\cite{Pearle} and,
independently, by Braunstein and Caves~\cite{BC}, based on earlier
work by Clauser, Horne, Shimony, and Holt~\cite{CHSH}.

The correlation measure is tailored to a specific bipartite setup
where measurements are carried out at two separate locations. One of
the measurements is specified by a parameter $A$ and has outcome
$X$. The other is specified by a parameter $B$ and has outcome $Y$. It
is furthermore assumed that the outcomes $X$ and $Y$ take values from
the binary set $\{0,1\}$ and that the parameters $A$ and $B$ are
labelled by elements from the sets $\{0,2,\ldots, 2N-2\}=:\cA_N$ and
$\{1,3,\ldots,2N-1\}=:\cB_N$, respectively, where $N$ is an
integer. The correlation measure, in the following denoted by $I_N$,
is then defined by
\begin{eqnarray*}
  I_N(P_{XY|AB}):=P(X=Y|A=0,B=2N-1)+\\
  \sum_{\genfrac{}{}{0pt}{}{a \in \cA_N, \, b \in \cB_N}{|a-b|=1}}\!\!\!\! P(X\neq Y|A=a,B=b)\, ,
\end{eqnarray*}
and depicted in Figure~\ref{fig:frust}.  Note that the measure only
depends on the conditional distribution $P_{XY|AB}$, and that stronger
correlations have a lower value of $I_N$.

\begin{figure}
\includegraphics[width=0.3\textwidth]{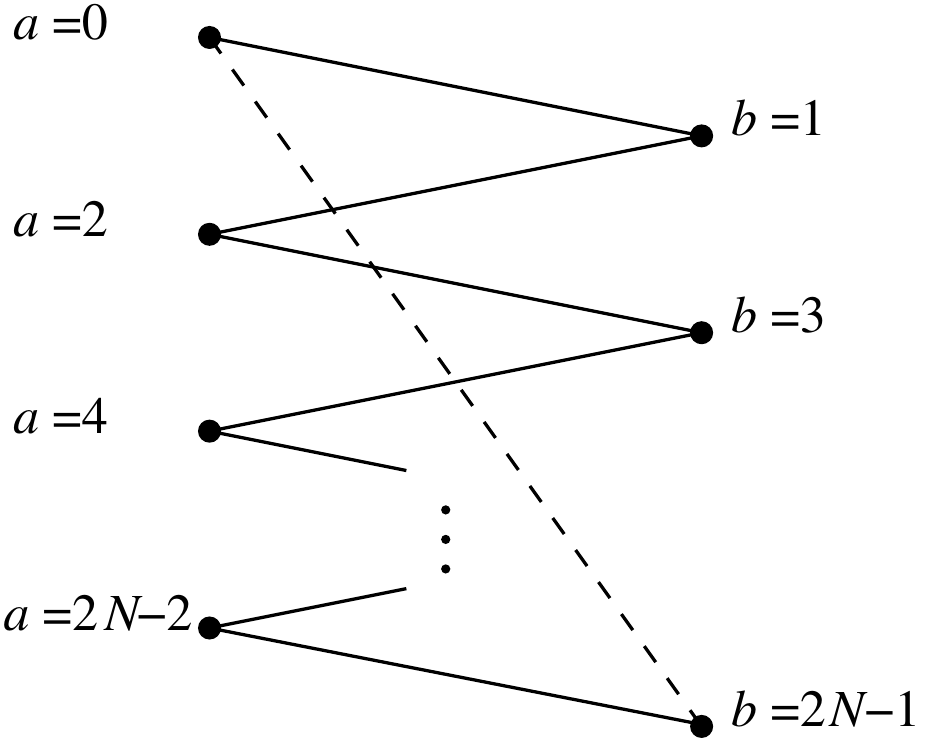}
\caption{{\sf \textbf{Illustration of the terms in the correlation
      measure $I_N$.}  This measure is defined as the sum of the
    probabilities of obtaining opposite outcomes when measuring two
    subsystems in neighbouring bases (depicted with the solid lines),
    and of obtaining the same outcomes for $a=0$, $b=2N-1$ (depicted
    with the dashed line).}}
\label{fig:frust}
\end{figure}

We will be particularly interested in the correlations that quantum
theory predicts for measurements on two maximally entangled two-level
systems. To specify these correlations, define
\[
\ket{\psi_2}:=\frac{1}{\sqrt{2}}(\ket{\up\up}+\ket{\dn\dn}) \ , 
\]
where $\{\ket{\up}, \ket{\dn}\}$ is an orthonormal basis. Furthermore,
let
$\ket{\theta}=\cos\frac{\theta}{2}\ket{\up}+\sin\frac{\theta}{2}\ket{\dn}$,
and take $E_x^a$ to be the projector onto $\ket{(\frac{a}{2N}+x)\pi}$
and, likewise, $F_y^b$ to be the projector onto
$\ket{(\frac{b}{2N}+y)\pi}$, as shown in Figure~\ref{fig:meas}.  We
then define $P^N_{XY|AB\psi_2}$ as the conditional distribution of
the outcomes of two separate quantum measurements, specified by
$\{E_x^a\}_x$ and $\{F_y^b\}_y$, respectively, applied to two separate
subsystems with joint state $\ket{\psi_2}$, i.e.,
\begin{align*}
  P^N_{XY|ab\psi_2}(x,y) := \bra{\psi_2} E_x^a \otimes F_y^b
  \ket{\psi_2} \ .
\end{align*}
It is easy to verify that the correlation strength, quantified with
the above correlation measure, $I_N$, equals
\begin{align} \label{eq_quantcorr}
  I_N(P^N_{XY|AB\psi_2})=2N\sin^2\frac{\pi}{4N} \leq \frac{\pi^2}{8N}\ .
\end{align}

\begin{figure}
\includegraphics[width=0.3\textwidth]{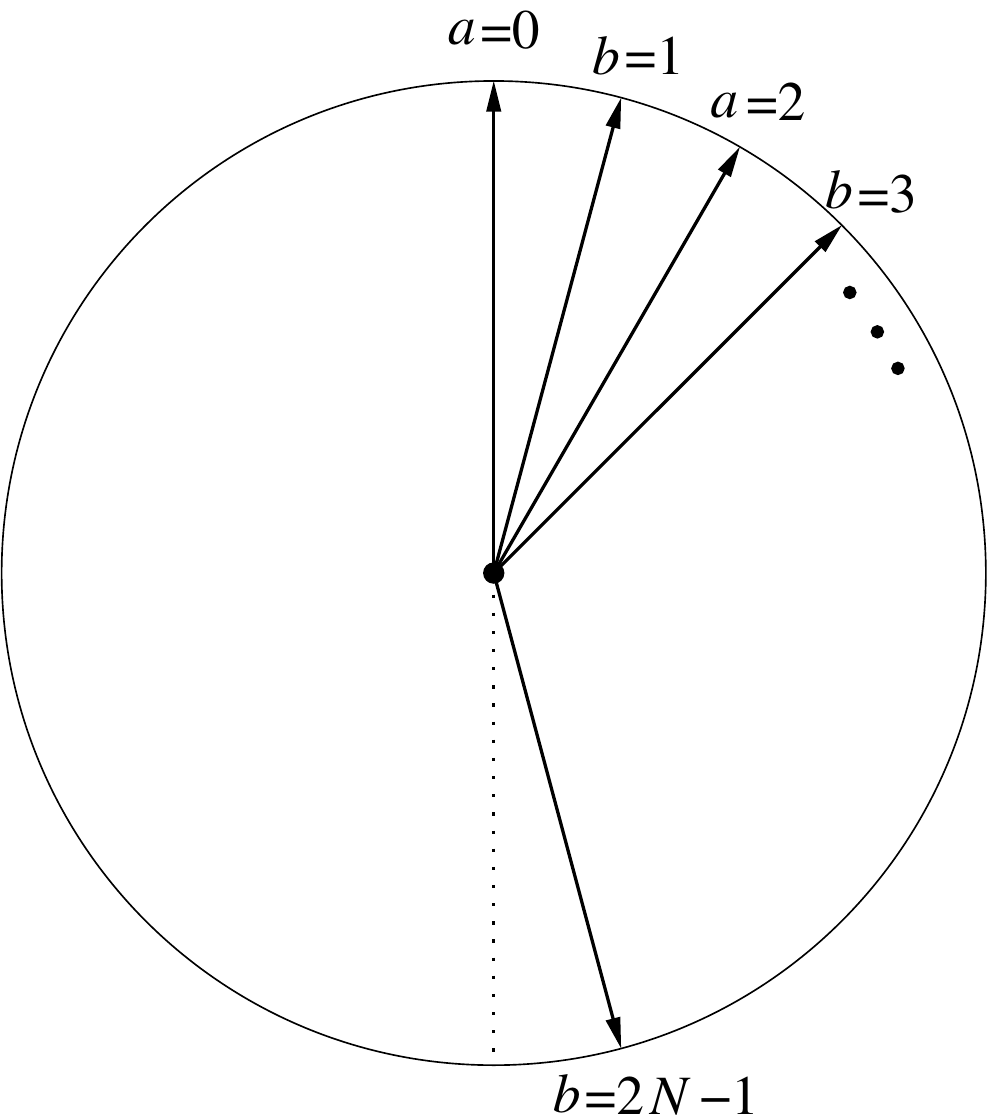}
\caption{{\sf \textbf{Depiction of the measurements used to achieve
      the quantum value of the correlation measure $I_N$.}  The circle
    represents the
    $\{\ket{\up},\frac{1}{\sqrt{2}}(\ket{\up}+\ket{\dn})\}$ plane of
    the Bloch sphere.  The arrows depict the Bloch vectors associated
    with the $0$ outcome (i.e.\ $E_0^a$ or $F_0^b$ are the projectors
    onto these states).  Those for the $1$ outcome lie in the opposite
    direction and are not depicted.  In the limit of large $N$, the
    measurements for neighbouring bases ($|a-b|=1$) are virtually
    identical and the outcomes are almost always perfectly correlated.
    Conversely, for $a=0$, $b=2N-1$ and large $N$, the measurements
    are virtually opposite of one another and the outcomes are almost
    always perfectly anti-correlated.}}
\label{fig:meas}
\end{figure}

\section{Quantum and alternative theories} \label{sec_theories}

The aim of this \chapter{} is to make statements about physical theories,
i.e., quantum theory as well as possible alternatives to it. However,
in order to derive our result, we do not need to provide a
comprehensive mathematical definition for the concept of a ``physical
theory''. Rather, it suffices to focus on one crucial feature that we
expect any theory to have, namely that it allows us to compute
predictions about values that can be observed (e.g., in an
experiment). These predictions, which need not be deterministic, are
generally based on certain parameters that characterize the
(experimental) setup, i.e., how it has been prepared (its initial
state), the evolution it undergoes, and which measurements are going
to be applied.

\subsection{Predictions of quantum theory} \label{sec_quantum}

In quantum theory, given the state, $\Psi$, of a system as well as a
specification of the measurement process, $A$, a prediction about an
experimentally observable value, $X$, can be obtained from Born's
rule. The state $\Psi$ may be given in the form of a density operator
on a Hilbert space $\cH$ and any measurement process $A=a$ can be
characterized by a \emph{Positive Operator Valued Measure (POVM)} on
$\cH$, i.e., a family of positive operators $\{E_x^a\}_x$ labelled by
the possible measurement outcomes $x \in \cX$ such that $\sum_x E_x^a
= \id_{\cH}$. (In this work, we assume for simplicity that the set
$\cX$ is finite.)

For our treatment, we will assume that any evolution of the system
prior to the measurement $\{E_x^a\}_x$ is already accounted for by its
quantum state, i.e, that $\Psi = \psi$ is the state of the system
directly before the measurement is applied.\footnote{Alternatively,
  one may work in the Heisenberg picture, for instance, and use the
  POVM to account for the evolution.} The predictions that quantum
theory makes about the measurement outcome $X$ can then be represented
as a conditional distribution $P_{X|A\Psi}$, which is given by
\begin{align} \label{eq_QMprediction} P_{X|a\psi}(x) = \tr(E^a_x \psi)
  \quad \forall\ x \in \cX \, .
\end{align}
We note that, by considering an extension of the Hilbert space $\cH$,
we may describe any quantum-mechanical measurement process
equivalently as a \emph{projective measurement}, i.e., one for which
the POVM $\{E^a_x\}_x$ consists of orthogonal
projectors.\footnote{According to Naimark's theorem, there exists a
  Hilbert space $\mathcal{\bar{H}}$ that contains $\cH$ as a subspace
  as well as orthogonal projectors $P^a_x$ in $\mathcal{\bar{H}}$ such
  that for each $x \in \cX$ the POVM element $E^a_x$ is the projection
  of $P^a_x$ into $\cH$.}  Furthermore, we call a set of POVMs
$\{E^a_x\}$ on $\cH$ \emph{tomographically complete} if the values
$P_{X|a\psi}(x)$ for all $a$ and $x$ are sufficient to determine
$\psi$ on $\cH$ uniquely.\footnote{An example of a tomographically
  complete set of projective POVMs in the case of a single qubit are
  the three POVMs whose elements are projectors onto (i) $\ket{\up}$
  and $\ket{\dn}$, (ii) $(\ket{\up}+\ket{\dn})/\sqrt{2}$ and
  $(\ket{\up}-\ket{\dn})/\sqrt{2}$, and (iii)
  $(\ket{\up}+i\ket{\dn})/\sqrt{2}$ and
  $(\ket{\up}-i\ket{\dn})/\sqrt{2}$.}

For later reference, we also note that, according to quantum theory,
any possible evolution of a quantum system, $S$, corresponds to a
unitary mapping on a larger state space (that may include the
environment of the system).  In the case of a measurement process,
this larger state space includes the measurement device,
$D$. Specifically, a projective measurement, say $\{E^a_x\}_x$, would
correspond to a unitary of the form
$$\ket{\psi} \mapsto \sum_x\sqrt{E_x^a} \ket{\psi}_S \ot \ket{x}_D ,$$
where $\{\ket{x}_D\}$ are orthonormal states of the measurement device
(and possibly also its environment) that encode the outcome.  The
outcome $X$ of the original measurement may then be recovered by a
subsequent projective measurement on $D$ in the basis $\{\ket{x}_D\}$.

\subsection{Predictions of alternative theories}

In an alternative theory, the measurement process $A$ with outcome
$X$, as described above in terms of the quantum formalism, may admit a
different description. This description could involve other
parameters, which we denote by $Z$ (one might think of $Z$ as the list
of all parameters used by the theory to describe the system's state
before the measurement $A$ is chosen).\footnote{In~\cite{CR_ext}, $Z$
  was modelled more generally as a system with input and output.  For
  simplicity, we ignore this higher level of generality in this work.}
For any values $A=a$ and $Z=z$ of these parameters, the theory
specifies a rule for computing the probability distribution,
$P_{X|az}$, for the measurement outcome $X$.  Hence, in the following,
if we want to make a statement about the predictive power of a given
theory,\footnote{When referring to the predictive power of a theory,
  we mean predictions based on the value $Z$.} it is sufficient to
consider the properties of the corresponding distributions $P_{X|az}$.

Since we want to use theories to make predictions, we usually think of
$Z$ as (in principle) learnable.  However, this is merely an
interpretive statement, and none of the conclusions of this work are
affected if $Z$ is instead thought of as forever hidden and hence
unlearnable in principle.  The only thing that changes in the latter
case is the interpretation of other statements.  In particular, one
may not want to call the condition $P_{XZ|AB}=P_{XZ|A}$, derived in
Section~\ref{sec:NS}, a ``no-signalling'' condition, or to speak about
``predictions'' made based on $Z$ if $Z$ is not learnable in
principle.

\subsection{Compatibility of predictions}

The predictions computed within two different theories (e.g., quantum
theory and an alternative theory) are generally not
identical. Nevertheless, they may be \emph{compatible} with each
other, in the following sense. Let $Z$ and $Z'$ be the parameters of
two different theories, and let their predictions (about the outcome
$X$ of a measurement $A$) be given by conditional probability
distributions $P_{X|AZ}$ and $P_{X|AZ'}$, respectively.\footnote{Note
  that the conditional probability distribution $P_{X|AZ}$ (and,
  similarly, $P_{X|AZ'}$) may be defined only for a restricted set of
  pairs $(a,z)$.}

\begin{definition} \label{def_compatible} $P_{X|AZ}$ and $P_{X|AZ'}$
  are said to be \emph{compatible} if there exists a conditional
  distribution $\bar{P}_{X Z Z'|A}$ such that\footnote{We require that
    both sides of the equalities are defined for the same pairs
    $(a,z)$ and $(a,z')$.}
  \begin{align*}
    P_{X|az} & = \sum_{z'} \bar{P}_{X Z'|a z}(\cdot, z') \ \
    \text{$\forall\ a,z$}\\
      P_{X|a z'} & = \sum_{z} \bar{P}_{X Z|a z'}(\cdot, z) \ \
      \text{$\forall\ a, z'$,}
  \end{align*}
  where the conditional distributions in the sums are derived from
  $\bar{P}_{X Z Z'|a}$.\footnote{That is, $\bar{P}_{X Z'|az}$ is given
    by
  \begin{align*}
    \bar{P}_{X Z'|az}(x, z') & = \bar{P}_{X Z Z' | a}(x, z,
    z')/\bar{P}_{Z|a}(z)  \quad \text{(if $\bar{P}_{Z|a}(z) > 0$)}
    \end{align*}
    where $\bar{P}_{Z|a}(z) = \sum_{x, z'} \bar{P}_{X Z Z'|a}(x, z,
    z')$, and likewise for $\bar{P}_{X Z|a z'}$.}
\end{definition}

To relate the definition back to the earlier example of the isotopes,
by way of illustration, the chemical element could be specified by
$Z$, and the particular isotope by $Z'$.  The relevant predictions are
then compatible in the above sense: since $Z'$ is a fine-graining of
$Z$ (i.e., $Z$ is uniquely determined by $Z'$), the second relation is
trivial, while the first recovers the non-isotopic predictions by
averaging over the different isotopes.

We will use this notion of compatibility to compare quantum theory to
alternative theories. For this, we let $Z' \equiv \Psi$ be the quantum
state of a system and consider the conditional distribution $P_{X|A
  \Psi}$ defined by~\eqref{eq_QMprediction}. An alternative theory
with predictions specified by $P_{X|A Z}$ (based on a parameter $Z$)
can then be considered compatible with quantum theory if there exists
a distribution $\bar{P}_{X Z \Psi|A}$ such that both $P_{X|A \Psi}$
and $P_{X|A Z}$ can be recovered from it (in the sense of the above
definition).

\subsection{Comparing the accuracy of predictions}

The predictive powers of different theories can be compared provided
the theories are mutually compatible.  The idea is that a theory with
predictions $P_{X|AZ}$ is \emph{at least as informative} as another
theory with predictions $P_{X|A Z'}$ if the latter can be obtained
from the former, i.e., if the parameter $Z'$ does not provide any
information beyond $Z$. This motivates the following definition.

\begin{definition}
  Let $P_{X|AZ}$ and $P_{X|AZ'}$ be compatible.  $P_{X|AZ}$ is said to
  be \emph{(at least) as informative as} $P_{X|AZ'}$ if there exists a
  conditional distribution $\bar{P}_{X Z Z'|A}$ as in
  Definition~\ref{def_compatible} such that
  \begin{align*}
    P_{X|az} =  \bar{P}_{X|azz'} \ \  \text{$\forall\ a,z,z'$
      s.t.\ $\bar{P}_{ZZ'|a}(z, z') > 0$ ,}
  \end{align*}
  where $\bar{P}_{X|azz'}$ and $\bar{P}_{ZZ'|a}$ are the conditional
  distributions derived from $\bar{P}_{XZZ'|A}$. 
\end{definition}

This can again be illustrated using the earlier example of the
isotopes. The theory that includes the information $Z'$ about the
particular isotope is of course at least as informative as the one
that only specifies the chemical element $Z$, but $Z$ is not as
informative as $Z'$.

We remark that quantum-mechanical predictions based on pure states are
generally more informative than those derived from mixed states. To
see this, imagine a system that is prepared in a pure state $\psi_C$
depending on a random bit $C$, and assume that a measurement with
outcome $X$ is performed. If $C$ is unknown, with $C=0$ and $C=1$
being equally likely, the distribution of $X$ is, according to quantum
theory, given by~\eqref{eq_QMprediction} with $\psi$ substituted by
the mixed state $\frac{1}{2} \psi_0 + \frac{1}{2} \psi_1$. However, if
we had access to $C$, we could use~\eqref{eq_QMprediction} with $\psi$
replaced by $\psi_C$, resulting in a more accurate prediction.

Clearly, when studying the question of whether there can be more
informative theories than quantum theory, we need to consider
specifications of states and measurement processes that are maximally
informative among all predictions that are possible within quantum
theory. Hence, following the above remark, we will restrict our
attention to quantum states that correspond to pure density operators
and to projective measurements.

\section{Freedom of choice} \label{sec_freedom}

As explained above, physical theories involve certain parameters, and
it is generally assumed (often implicitly) that these can be chosen
freely. Quantum mechanics, for instance, allows us to compute the
probabilities of a measurement outcome $X$ depending on the system's
state $\Psi$ as well as a description of the measurement process, $A$,
and our understanding is that these parameters can in principle be
chosen freely (e.g., by an experimenter carrying out a measurement of
her choice).  In fact, one may argue that a description of nature that
does not involve any such choices|thereby not allowing us to compute
conclusions for different initial conditions|cannot be reasonably
termed a theory~\cite{Bell_free}.

It is worth noting that by assuming free choice, we are not making any
metaphysical assertion that the real world contains, say, agents with
free will, or anything of that sort.  Instead, allowing free choice is
a property that we require of a theory.  In essence, it means that the
theory gives predictions for all possible values of the free
parameters, and furthermore, that it does so no matter what happened
elsewhere in the theory.  Without such an assumption, depending on
other events described by the theory, certain values of the `free'
parameters could be unavailable, in the sense that the theory would
not be able to predict a response to them.

In this section, we specify what we mean by such free choices. The
idea is that, for a given theory, the statement that a parameter of
the theory, say $A$, is considered \emph{free} is equivalent to saying
that $A$ is uncorrelated with all values (described by the theory)
that are outside the future of $A$.  For this definition to make sense
mathematically, we need to establish a notion of \emph{future}. We do
this by introducing a \emph{causal order}, i.e., a (partial) ordering
of events.  We stress, however, that the causal order is only used to
define free choice and plays no further part in the argument.  \footnote{In
particular, we do not assume~\emph{local causality} within the
specified causal order.}

\subsection{Causal order}

Let $\Gamma$ be the set of all parameters required for the description
of an experiment within a given theory.  In particular, $\Gamma$ may
contain variables that specify the (joint) state in which the relevant
physical systems have been prepared (in the following usually denoted
by $\Psi$ for quantum theory and by $Z$ for more general theories),
the choice of measurements (denoted $A$ and $B$), as well as the
measurement outcomes (denoted $X$ and $Y$).  For any such set of
variables $\Gamma$, we can define a causal order $\bef$ as follows.

\begin{definition}
  A \emph{causal order} $\bef$ for $\Gamma$ is a preorder
  relation\footnote{That is, $\bef$ is a binary relation on the set
    $\Gamma$ that is reflexive (i.e., $A \bef A$) and transitive
    (i.e., $Z \bef A$ and $A \bef X$ imply $Z \bef X$).} on $\Gamma$.
\end{definition}
If $A\bef X$, we say that \emph{$X$ is in the (causal) future of $A$},
and if this doesn't hold, we write $A\nbef X$.  These relations can be
conveniently specified by a diagram (see Figure~\ref{fig:cause} for an
example).  Note that the causal order~$\bef$ should not be interpreted
as specifying actual causal dependencies\footnote{I.e., $A\bef X$ is
  not meant to imply that there is necessarily a physical process such
  that changing $A$ imposes a change of $X$.}, but instead indicates
that such causal dependencies are not precluded (by the theory).

% To understand the physical relevance of the statements in this work,
% it is useful to interpret the relation $A\bef X$ as ``$A$ \emph{can}
% be the cause of $X$''.  We stress that this is not meant to imply that
% there is an actual physical process such that changing $A$ imposes a
% change of $X$, but rather that the existence of such a process is not
% precluded (by the theory).  Conversely, if $A\bef X$ does not hold, we
% write $A\nbef X$, which can be interpreted as \emph{$A$ cannot be the
%   cause of $X$}.

% $A\bef X$ can be interpreted as \emph{$X$ is in the causal future of
%   $A$} or \emph{$A$ is before $X$}.  If this does not hold, we write
% $A\nbef X$, which can be interpreted as \emph{$A$ is not before $X$}
% (note that this is not the same as $A$ is after $X$).  The relation
% $\bef$ can be conveniently specified by a diagram (see
% Figure~\ref{fig:cause} for an example).  The causal order $\bef$
% should not be interpreted as specifying actual causal
% dependencies\footnote{I.e., $A\bef X$ is not meant to imply that there
%   is an actual physical process such that changing $A$ imposes a
%   change of $X$.}, but instead indicates that such causal dependencies
% are not precluded (by the theory).

% This is analogous to
% the light cone structure in relativistic space time.  An event has the
% potential to cause others in its future light cone, but doesn't
% necessarily do so.

A typical|but for the following considerations not
necessary|requirement on a causal order is that it be compatible with
relativistic space time. Consider, for example, an experiment where a
parameter $A$ is chosen at a given space time point $\textbf{r}_A$ and
where a measurement outcome $X$ is observed at another space time
point $\textbf{r}_X$. One would then naturally demand that $A\bef X$
if and only if $\textbf{r}_X$ lies in the future light cone of
$\textbf{r}_A$. This captures the idea that the choice $A$ is made at
an earlier time than the observation of $X$, with respect to any
reference frame.

\subsection{Free random variables} \label{sec_free}

To define the notion of a ``free choice'', we consider a set $\Gamma$
of RVs equipped with a causal order. (As above, $\Gamma$ should be
thought of as the set of all parameters relevant for the description
of an experiment within a given theory.)

\begin{definition}\label{def:free}
  We say that $A \in \Gamma$ is \emph{free} if
  \begin{align*}
    P_{A\Gamma_A} = P_A \times P_{\Gamma_A}
  \end{align*}
  holds, where $\Gamma_A$ is the set of all RVs $X \in \Gamma$ such
  that $A\nbef X$.\footnote{By definition, the set $\Gamma_A$ also
    excludes $A$.}
\end{definition}

Obviously, whether a variable from the set $\Gamma$ is considered free
depends on the causal order that we impose.  If the causal order is taken
to be the one induced by relativistic space time (see the description
above), then this definition coincides with the notion of a \emph{free
  variable} as used by
Bell~\cite{Bell_free}.\footnote{In~\cite{Bell_free}, Bell discusses
  the assumption that the settings of instruments are \emph{free
    variables}, which he characterizes as follows: ``For me this means
  that the values of such variables have implications only in their
  future light cones.''}  We remark that both standard quantum theory
and classical theory in relativistic space time allow for free choices
within such a causal order.

\begin{figure}
\includegraphics[width=0.25\textwidth]{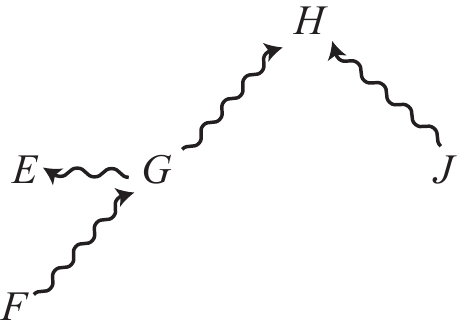}
\caption{{\sf \textbf{Free choice and causal order.} An arbitrary
    causal order is depicted for random variables $E$, $F$, $G$, $H$ and
    $J$. The arrows correspond to the relation $\bef$. For
    example, $G$ lies in the future of $F$, i.e., $F \bef G$, but not
    of $J$, i.e., $J\nbef G$.  Because of the transitivity property,
    it follows that $F\bef E$, for example.  In this setting we would
    say that, for instance, $G$ is free if it is uncorrelated with $F$
    and $J$, i.e., $P_{GFJ}=P_G\times P_{FJ}$.}}
\label{fig:cause}
\end{figure}

\section{Constraints on theories compatible with quantum
  theory} \label{sec_constraints}

We discuss here the implication of some well-known results to our main
question, whether an extension of quantum theory can have improved
predictive power.  Although they were not asking the same question,
the works of Bell~\cite{Bell} and Leggett~\cite{Leggett} can be
adapted to give constraints on such higher theories, and hence give
special cases of the general theorem presented in
Section~\ref{sec_extended}, which excludes all alternative theories
whose predictions are more informative than quantum theory.

% The debate about whether quantum theory could be replaced by a higher
% (possibly deterministic) theory has a long history (see also the
% introductory section). The common feature of all proposed higher
% theories is that they would make more informative predictions than
% quantum theory. Here, we review some well-known results that impose
% constraints on such higher theories. We note that these constraints
% can be seen as special cases of the general theorem presented in
% Section~\ref{sec_extended}, which excludes all alternative theories
% whose predictions are more informative than quantum theory.

\subsection{Bipartite setup} \label{sec_bipartite}

The statements described below refer to a bipartite setup which
involves two separate measurements, specified by parameters $A$ and
$B$, and with outcomes $X$ and $Y$, respectively.  As before, we
consider a theory that allows us to compute predictions about these
measurements, based on a parameter (or list of parameters)
$Z$. Furthermore, in order to define free choices, we need to specify
a causal order.  The technical claims described in this section can be
applied to any causal order that satisfies the following conditions:
\begin{itemize}
\item [(i)] $A \bef X$ and $B \bef Y$;
\item [(ii)] $A\nbef Z$ and $B\nbef Z$;
\item [(iii)] $A\nbef Y$ and $B\nbef X$.
\end{itemize}
Condition~(i) corresponds to the requirement that the measurement is
specified before its outcome is obtained.  Condition~(ii) captures the
fact that the parameters of the theory, $Z$, on which the predictions
are based, should not only become available after the measurement
process is started. This assumption can be considered necessary in
order to reasonably talk about ``predictions''.  Finally,
Condition~(iii) demands that the arrangement of the two measurements
should be such that neither of them lies in the future of the
other. (Note that, assuming a relativistic space time structure, this
would correspond to a setup where the measurements are space-like
separated.)  Together, the three conditions imply a causal order in
which $A$ is considered free if $P_{ABYZ}=P_A\times P_{BYZ}$, and
likewise for~$B$.  The causal orders respecting (i)--(iii) are
illustrated in Figure~\ref{fig:chronology}.

\begin{figure}
\includegraphics[width=0.2\textwidth]{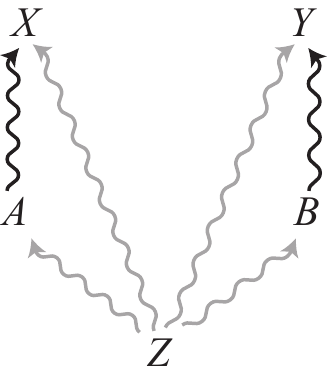}
\caption{{\sf \textbf{The causal orders for which our argument
      applies.}  We consider a setup with two separate measurements,
    one depending on a choice $A$ with outcome $X$, and the other with
    choice $B$ and outcome $Y$.  Moreover, $Z$ denotes all extra
    parameters that may be used to make predictions about the
    outcomes. The figure illustrates all of the causal orders
    compatible with our requirements~(i)--(iii).  The black arrows
    originating from $A$ and $B$ are required, while each of the grey
    arrows originating from $Z$ is
    optional.%(Note that, due to transitivity, there is no
%      need to include both grey arrows on the same side, thus there
%      are 9 distinct causal orders in total.)
  }}
\label{fig:chronology}
\end{figure}

\subsection{Local deterministic theories} \label{sec_deterministic}

Local deterministic theories were introduced in the work of
Bell~\cite{Bell}, and,  %Determinism means that the outcomes of measurements
%can be predicted with certainty, given access to the parameters of the
%theory, $Z$ (sometimes termed ``hidden variables'').  Local (or
%\emph{local causality}) refers to the additional requirement that the
%predictions only depend on ``local'' parameters within a (classical)
%causal model (the precise mathematical requirements of local causality
%are stated in the next subsection).
within the bipartite setup described above, are ones for which all
conditional probabilities $P_{X|az}(x)$ and $P_{Y|bz}(y)$ are equal to
either $0$ or $1$.

% means that the
% measurement outcome $X$ only depends on the choice of measurement $A$
% as well as $Z$, and, similarly, $Y$ only depends on $B$ and $Z$.
% Determinism and local causality together imply that all conditional
% probabilities $P_{X|az}(x)$ and $P_{Y|bz}(y)$ must be equal to either
% $0$ or $1$. This property is also called \emph{local determinism}.

It follows using essentially the same argument used to prove Bell's
theorem that no locally deterministic theory can reproduce the
predictions of quantum theory.  To show this, it is sufficient to
consider the correlations $P^N_{XY|AB\psi_2}$ that quantum theory
predicts for the measurements on the maximally entangled state
$\ket{\psi_2}$ defined in Section~\ref{sec_CB}.  We state this as the
following theorem.

\begin{theorem}[No higher theories are locally
  deterministic] \label{thm_Bell} Let $A$, $B$, $X$, $Y$ and $Z$ be
  RVs. Then at least one of the following cannot hold:
  \begin{itemize}
  \item \emph{Freedom of choice:\footnote{The freedom of choice
        assumption is often not mentioned explicitly, but its
        necessity has been stressed by Bell in later
        work~\cite{Bell_free}.}} $A$ and $B$ are free with respect to
    any of the causal orders depicted in
    Figure~\ref{fig:chronology};%\footnote{Note that in the context of
%      Bell's theorem, the chronology is usually taken to be that of
%      Figure~\ref{fig:chronology}(a) or (e).}
 \item \emph{Compatibility with quantum theory:} $P_{X Y | A B Z}$ is
   compatible with the predictions $P^N_{XY|AB\psi_2}$ of quantum
    theory (for some $N \geq 2$);
  \item \emph{Local determinism:} 
    \begin{align*}
      P_{X|a z}(x) & \in \{0,1\} \quad \text{$\forall\ a, z $ s.t.\
        $P_{AZ}(a,z) > 0$} \\
      P_{Y|b z}(y) & \in \{0,1\} \quad \text{$\forall \ b, z $ s.t.\
        $P_{BZ}(b,z) > 0$ .}
   \end{align*}
\end{itemize}
\end{theorem}

To prove this theorem, we use the correlation measure $I_N$ defined in
Section~\ref{sec_CB}. The central idea is to show that, under the free
choice assumption, all correlations explained by a locally
deterministic model satisfy the inequality $I_2 \geq 1$, which
corresponds to the CHSH inequality~\cite{CHSH}. (The free choice
assumption ensures that $P_{AB|z}$ has full support for each $z$, and
hence that the conditional distributions $P_{X|az}$ and $P_{Y|bz}$
are well defined for any $a$, $b$, and $z$.) The assertion then
follows from the fact that $I_N(P^N_{XY|AB\psi_2}) = 2-\sqrt{2} < 1$
for $N=2$ (see~Eq.~\ref{eq_quantcorr}).

\subsection{Stochastic local causal theories}

In his later work, Bell dropped the assumption of determinism and
considered more general stochastic models.  He adopted the following
definition of locality called \emph{local causality}, which leads to
the relation $P_{XY|ABZ}=P_{X|AZ}P_{Y|BZ}$~\cite{Bell_nouvelle}.
Expanding the left hand side using Bayes' rule, this can be broken
down into four separate relations, $P_{X|ABZ}=P_{X|AZ}$,
$P_{Y|ABZ}=P_{Y|BZ}$, $P_{X|ABYZ}=P_{X|ABZ}$ and
$P_{Y|ABXZ}=P_{Y|ABZ}$.  The first two of these have sometimes been
termed \emph{parameter independence} and imply that, even given access
to $Z$, there cannot be signalling between the two measurement
processes.

The last two conditions have been termed \emph{outcome
  independence}. They do not have an obvious operational significance
(such as no-signalling), and do not in general hold for the theories
we consider in this work. We note, however, that they are
automatically satisfied in any deterministic model, where each of the
outcomes $X$ and $Y$ is a function of $A$, $B$, and $Z$.  Conversely,
as we argue below, if a theory is locally causal then the predictions
it makes about the outcomes of measurements on the entangled state
$\ket{\psi_2}=\frac{1}{\sqrt{2}}(\ket{\up\up}+\ket{\dn\dn})$ are
necessarily deterministic.  (This is the essence of the EPR
argument~\cite{EPR}.)

To see this, note that for any projective measurement (specified by
$A=a$) applied to the first part of $\ket{\psi_2}$, there exists
another projective measurement (specified by $B = b_a$) on the second
part such that the outcomes are perfectly correlated. For example, if
$A=a$ corresponds to the POVM $\{\ketbra{\up}{\up},
\ketbra{\dn}{\dn}\}$, and if we choose $B=b_a$ such that it
corresponds to the same POVM, then $P_{XY|a
  b_a}(0,0)=P_{XY|ab_a}(1,1)=\frac{1}{2}$. This means that $X$ is
determined by $Y$, i.e., $P_{X|ab_ayz}(x)=\delta_{x,y} \in \{0,1\}$
for all $a$, $x$, $y$ and $z$. Applying now the conditions of local
causality, we obtain $P_{X|abyz}(x) = P_{X|az}(x) \in \{0,1\}$, which
corresponds to the assumption of local determinism. Hence, there is an
analogue of Theorem~\ref{thm_Bell} in which the local determinism
condition is weakened to Bell's local causality condition.

We remark that, as we shall see below (Lemma~\ref{lem:1}), the freedom
of choice assumption implies parameter independence, but is not strong
enough to imply local causality, since it doesn't imply outcome
independence.

\subsection{Leggett-type theories} \label{sec_Leggett}

In~\cite{Leggett}, Leggett introduced what he calls a ``non-local
hidden variable'' model, which attempts to give an explanation of
quantum correlations that is partly local and partly non-local.  The
presence of non-local hidden variables in his model leads to an
incompatibility with the free choice assumption, and hence Leggett's
model is not a higher theory in the sense of the present work.
However, since the behaviour of the non-local variables is not
specified in Leggett's model, we can consider a slightly modified
version in which they are ignored (henceforth, when we speak about
Leggett's model, we refer to the local part of it).  The model is then
compatible with our notion of free choice, and offers improved
predictive power for measurements on maximally entangled particles. We
note that the model is not a full-fledged theory, as it only specifies
how the outcomes of spin measurements are obtained.

Leggett's model is based on the idea of assigning to each spin
particle a three-dimensional vector (in addition to its quantum
mechanical state). In particular, if we consider two spin particles,
each measured on one side within the bipartite setup described above,
we need to specify two such vectors, denoted ${\bf u}$ and ${\bf v}$,
respectively.  To connect this to our general discussion, we may think
of these vectors as part of $Z$, i.e., $Z$ takes as values pairs
$({\bf u}, {\bf v})$. As above, we denote the choice of measurement on
each side by $A$ and $B$. Restricting to projective spin measurements,
the two choices may be labelled by three-dimensional vectors, denoted
${\bf a}$ and ${\bf b}$, respectively, indicating their orientation in
space (see, for example,~\cite{BBGKLLS} for more details). The
predictions for the measurement outcomes $X$ and $Y$, as prescribed by
Leggett's model, are then given by
\begin{eqnarray}\label{eq:leg1}
P_{X|{\bf a u v}}(x)&=&\frac{1}{2}(1+(-1)^x{\bf a}\cdot{\bf u})\\
P_{Y|{\bf b u v}}(y)&=&\frac{1}{2}(1+(-1)^y{\bf b}\cdot{\bf v})\, .\label{eq:leg2}
\end{eqnarray}

In order to completely define the model, one would also need to assign
probabilities to all possible values $Z = ({\bf u}, {\bf v})$, i.e.,
specify a probability distribution $P_Z$ (which, in general, depends
on the quantum state). However, the following theorem, which is a
corollary of results in~\cite{Leggett,GPKBZAZ,BBGKLLS,ColbeckRenner},
implies that there is no such assignment for which Leggett's model can
be made compatible with quantum theory.

\begin{theorem}[No higher theories obey the Leggett
  conditions] \label{thm_Leggett} Let $A$, $B$, $X$, $Y$ and $Z$ be
  RVs. Then there exists a quantum distribution, $P_{XY|AB\psi_2}$,
  such that at least one of the following cannot hold:
  \begin{itemize}
  \item \emph{Freedom of choice:} $A$ and $B$ are free with respect to
    any of the causal orders depicted in Figure~\ref{fig:chronology};
  \item \emph{Compatibility with quantum theory:} $P_{XY|ABZ}$ is
    compatible with $P_{XY|AB\psi_2}$;
  \item \emph{Leggett rule:} $P_{XY|ABZ}$ satisfies
    Eqs.~\ref{eq:leg1} and~\ref{eq:leg2} for all values $A=a$, $B=b$,
    and $Z = ({\bf u}, {\bf v})$.
  \end{itemize}
\end{theorem}

We will not give a proof of this theorem here, since it follows from
the more general results presented in the next section.  To see this,
it is sufficient to observe that, when measuring the entangled state
$\ket{\psi_2}=\frac{1}{\sqrt{2}}(\ket{\up\up}+\ket{\dn\dn})$, for
instance, quantum theory prescribes that $P_{X|{\bf
    a}}(x)=\frac{1}{2}$, independently of the orientation ${\bf a}$ of
the measurement. Conversely, Leggett's model predicts a non-uniform
distribution whenever the measurement orientation ${\bf a}$ is not
orthogonal to the vector ${\bf u}$. The Leggett model is therefore
more informative than quantum theory, and hence excluded by
Lemma~\ref{thm_entanglementpredict} (as well as the more general
Theorem~\ref{thm_main}) below.

\subsection{Other Constraints}

Here we summarize a few other known constraints on theories compatible
with quantum mechanics.  One of the first results in this direction
was that the quantum outcomes cannot be predetermined within a
non-contextual model~\cite{KS,Bell_KS}.  In such a model, one assumes
the existence of a map from the set of projectors to the set $\{0,1\}$
such that for every set of projectors that constitute a POVM, only one
member of that set is mapped to $1$ (the element that maps to $1$ is
interpreted as the outcome that will occur if a measurement described
by that POVM is carried out).  Such a model is \emph{non-contextual}
in that whether or not a particular outcome occurs depends only on the
individual projectors, and not on the set of projectors making up the
POVM.  The Bell-Kochen-Specker theorem~\cite{KS,Bell_KS} implies that
no such assignment can exist if the Hilbert space dimension is at
least~$3$.

Hardy~\cite{Hardy_ontbag} later showed that within any extended
theory, an infinite number of underlying states are required, even to
describe a single qubit, and Montina~\cite{Montina3,Montina} proved,
under the assumption of Markovian dynamics, that the number of real
parameters that an extended theory needs to characterize a state in
Hilbert space dimension $N$ is at least $2N-2$ (the same as the number
of parameters needed to specify a pure quantum state up to global
phase).

In addition, a claim in the same spirit as our non-extendibility
theorem (presented in the next Section) has been obtained
recently~\cite{ChenMontina} under the assumption of measurement
non-contextuality, introduced in~\cite{Spekkens_context}.

\section{The non-extendibility theorem} \label{sec_extended}

This section is devoted to the key result of this \chapter, asserting
that quantum theory is maximally informative.  Stated informally, we
make the following claim, first made in~\cite{CR_ext}.
\begin{claim}\label{claim1}
  No alternative theory that is compatible with quantum theory and
  allows for free choice (with respect to the discussed causal orders)
  can give improved predictions.
\end{claim}

The main technical statement is a generalization of the theorems
discussed in the previous section.  The setup is broadly the same, but
instead of the condition that the higher theory remains compatible
with quantum theory for measurements on maximally entangled states, we
require this for a wider class of states.  Furthermore, rather than
considering theories that satisfy local determinism or the Leggett
rule, the claim is about arbitrary theories that make improved
predictions.

The main technical theorem is as follows (this should be read as a
purely mathematical statement about bipartite pure states, whose
significance to the extendibility of quantum theory will be explained
subsequently). 
\begin{theorem}\label{thm_main} Let $\ket{\phi}_{S D}$ be a pure state and
  let $\{\ket{\hat{y}}_D\}$ be a Schmidt basis on $D$. Then there
  exists a state $\ket{\Gamma}_{\tilde{S} \tilde{D}}$ and local POVMs
  $\{E_x^a\}$ and $\{F_y^b\}$ on $S \tilde{S}$ and $D\tilde{D}$,
  respectively, with $F_y^{b_0}=\proj{\hat{y}}_D\ot\id_{\tilde{D}}$
  for some $b=b_0$, such that, for any RVs $A$, $B$, $X$, $Y$ and $Z$,
  at least one of the following cannot hold:\footnote{Strictly
    speaking, the entangled state and POVMs should be sequences of
    entangled states and POVMs, for which the maximum improvement in
    the prediction tends to $0$ (c.f.\
    Lemma~\ref{thm_entanglementpredict}).}
  \begin{itemize}
  \item \emph{Freedom of choice:} $A$ and $B$ are free with respect to
    any of the causal orders depicted in
    Figure~\ref{fig:chronology};%\footnote{More generally, the statement
%      holds for any causal order that satisfies the three conditions
%      given in Section~\ref{sec_bipartite}.}
  \item \emph{Compatibility with quantum theory:} $P_{XY|ABZ}$ is
    compatible with the prediction $P_{XY|AB(\phi \otimes \Gamma)}$ of
    quantum theory for the measurements $\{E_x^a\}$ and $\{F_y^b\}$ on
    $\ket{\phi}_{S D} \otimes \ket{\Gamma}_{\tilde{S}
      \tilde{D}}$.\footnote{Formally, $P_{XY|AB(\phi \otimes \Gamma)}$
      is given by \begin{align*} P_{XY|ab(\phi \otimes \Gamma)}(x,y) =
        \tr ((E_x^a\ot F_b^y)\proj{\phi \otimes \Gamma}) \, .
  \end{align*}}
\item \emph{Improved predictions:} $P_{Y|b_0\phi}$ is not as
  informative as $P_{Y|b_0 Z}$.
  \end{itemize}
\end{theorem}

To understand the implications of this theorem, consider a fixed
measurement $\hat{a}$ on a system $S$.  Assume that, according to
quantum theory, the system (before the measurement) is in a pure
state, denoted $\psi$, and that the measurement corresponds to a
projective POVM, $\{\hat{E}^{\hat{a}}_{\hat{x}}\}$. Quantum theory
then gives a probabilistic prediction $P_{\hat{X}| \hat{a} \psi}$ for
the measurement outcome $\hat{X}$, which depends on $\psi$ and
$\{\hat{E}^{\hat{a}}_{\hat{x}}\}$ (see Eq.~\ref{eq_QMprediction}).
Our aim is to compare this quantum-mechanical prediction with the
prediction $P_{\hat{X}|\hat{a} Z}$ that may be obtained by an
alternative theory, whose parameters we denote by $Z$.

In order to relate this to the theorem, let us assume that the freedom
of choice condition holds, from which it follows that the alternative
theory is no-signalling (i.e., it would be impossible to signal in the
higher theory, even knowing the parameter $Z$).  We then consider the
joint state of the measured system, $S$, and the measurement device,
$D$, after the measurement $\hat{a}$. Following the discussion in
Section~\ref{sec_theories}, according to quantum theory, this state
can be assumed to have the form
\begin{align} \label{eq_entaftermeas}
  \ket{\phi}_{S  D} = \sum_{\hat{x}} \sqrt{\hat{E}^{\hat{a}}_{\hat{x}}}\ket{\psi}_S \ot
  \ket{\hat{x}}_D \, \ .
\end{align}
Note that the POVM $\{F_y^{b_0}\}$ defined by Theorem~\ref{thm_main}
corresponds to a measurement of $D$ in the basis
$\{\ket{\hat{x}}_D\}$. The outcome, $Y$, of this measurement can
therefore be seen as a copy of the outcome $\hat{X}$ of the original
measurement, specified by $\{\hat{E}^{\hat{a}}_{\hat{x}}\}$. In
particular, the prediction that any theory compatible with quantum
theory makes about $Y$ must be identical to the prediction it makes
about $\hat{X}$, i.e., we have
\begin{align*}
  P_{\hat{X}|\hat{a} \psi} & = P_{Y|b_0 \phi}  \\
  P_{\hat{X}|\hat{a} Z} & = P_{Y|b_0 Z} \ .
\end{align*}
(Note that because the free choice assumption implies that the
alternative theory is no-signalling, the prediction the alternative
theory makes about $Y$ does not depend on $\hat{a}$, for example.)

We now apply Theorem~\ref{thm_main} to $\ket{\phi}_{S D}$. If we
assume that the alternative theory, in addition to being compatible
with quantum theory, satisfies the freedom of choice assumption, then
the third condition of the theorem cannot hold, i.e., $P_{Y|b_0 \phi}$
is as informative as $P_{Y|b_0 Z}$.  Using the above
identities, this directly carries over to the original measurement
$\hat{a}$, i.e., the quantum-mechanical prediction $P_{\hat{X}|\hat{a}
  \psi}$ is as informative as the prediction $P_{\hat{X}|\hat{a} Z}$
of the alternative theory.  We hence establish Claim~\ref{claim1}.

\section{Proof of Theorem~\ref{thm_main}} \label{sec_proof}

The theorem follows from three statements, which we formulate and
prove separately. An overview of the argument is as follows.  We
consider the previously introduced bipartite scenario and any of the
causal orders depicted in Figure~\ref{fig:chronology}.  We begin by
showing that free choice with respect to this causal order implies that
the alternative theory is no-signalling (see Lemma~\ref{lem:1}).  In
the second part of the argument, we show that for measurements on
maximally entangled states, if quantum theory is correct, no higher
theory can give improved predictions about the outcomes (see
Lemma~\ref{thm_entanglementpredict}).  In the final part of the
argument, we generalize this to measurements on an arbitrary bipartite
entangled state.  More precisely, we show that for any such state,
there exist local measurements that generate correlations arbitrarily
close to those generated by $r$ maximally entangled states for some
sufficiently large integer $r$.  Hence, from the second part of the
argument, these measurements can have no improved predictions.

\subsection{Part~I: No-signalling from free choice}\label{sec:NS}

In this part, we show that if $A$ and $B$ are free choices with
respect to one of the given causal orders, then there is no signalling
within the alternative theory (i.e.\ no signalling even given access
to $Z$).\footnote{As explained in Section~\ref{sec_theories}, the
  interpretation of this as ``no-signalling'' may change if $Z$ is
  thought of as in principle unlearnable.  However, we stress that the
  same mathematical conditions remain in that case too.}

\begin{lemma}\label{lem:1}
  The freedom of choice assumption implies $P_{XZ|AB}=P_{XZ|A}$ and
  $P_{YZ|AB}=P_{YZ|B}$.
\end{lemma}

\begin{proof}
  That $A$ is free within the specified causal order implies
  $P_{A|BYZ}=P_A$ and hence
  \begin{align*}
    P_{YZA|B}&= P_{YZ|B}\times P_{A|BYZ}=P_A\times P_{YZ|B}\,\text{, and}\\
    P_{YZA|B}&= P_{A|B}\times P_{YZ|AB}=P_A\times P_{YZ|AB}\, .
  \end{align*}
  We therefore have $P_{YZ|AB}=P_{YZ|B}$.  The relation
  $P_{XZ|AB}=P_{XZ|A}$ follows by symmetry.
\end{proof}

\subsection{Part~II: Non-extendibility for measurements on maximally
  entangled states} 

In the second part of the argument, we show that the claim holds for
particular measurements on maximally entangled pairs of qubits. The
proof uses the correlation measure $I_N$ introduced in
Section~\ref{sec_CB}. The following lemma shows that this measure,
applied to a distribution $P_{X Y|A B}$, gives a bound on how well any
additional information, $Z$, can be correlated to the outcome $X$.
Note that the lemma is independent of quantum theory and is simply a
property of probability distributions.

\begin{lemma}\label{lem:2}
  Let $P_{XYZ|AB}$ be a distribution that obeys $P_{XZ|AB}=P_{XZ|A}$
  and $P_{YZ|AB}=P_{YZ|B}$.  Then, for all $a$ and $b$, we have
\begin{equation}
  \bigl\langle D(P_{X|abz},P_{\bar{X}}) \bigr\rangle_z \leq\frac{1}{2}I_N(P_{XY|AB}) \, ,\label{eq:markov2}
\end{equation}
where $\langle \cdot \rangle_z$ denotes the average over the values of
$Z$ (distributed according to $P_{Z|ab}$), and $P_{\bar{X}}$ denotes
the uniform distribution on $X$.
\end{lemma}

The proof is based on an argument given in~\cite{CR_ext}, which
develops results of~\cite{BHK},~\cite{BKP} and~\cite{ColbeckRenner}.

\begin{proof}
  We first consider the quantity $I_N$ evaluated for the conditional
  distribution $P_{XY|ABz} = P_{XY|ABZ}(\cdot, \cdot|\cdot, \cdot,
  z)$, for any fixed $z$. The idea is to use this quantity to bound
  the variational distance between the conditional distribution
  $P_{X|az}$ and its negation, $1-P_{X|az}$, which corresponds to the
  distribution of $X$ if its values are interchanged.  If this
  distance is small, it follows that the distribution $P_{X|az}$ is
  roughly uniform. Because this holds for any $Z=z$, $X$ must be
  independent of $Z$.

It is first worth noting that the conditions of the lemma
($P_{XZ|AB}=P_{XZ|A}$ and $P_{YZ|AB}=P_{YZ|B}$) imply
$P_{X|ABZ}=P_{X|AZ}$ and $P_{Y|ABZ}=P_{Y|BZ}$ respectively, and
together imply $P_{Z|AB}=P_{Z}$.

Let $P_{\bar{X}}$ be the uniform distribution on $X$. For $a_0 := 0$,
$b_0:= 2N-1$, we have
\begin{align}
&I_N(P_{XY|ABz})\nonumber\\
&=P(X=Y|a_0,b_0,z)+\!\!\!\!\sum_{\genfrac{}{}{0pt}{}{a\in \cA_N, \, b \in \cB_N}{|a-b|=1}}\!\!\!\! P(X\neq Y|a,b,z) \nonumber   \\
&\geq D(1-P_{X|a_0 b_0 z},P_{Y|a_0 b_0 z}) +\!\!\!\!
\sum_{\genfrac{}{}{0pt}{}{a\in \cA_N, \, b \in
    \cB_N}{|a-b|=1}}\!\!\!\! D(P_{X|abz},P_{Y|abz})
\nonumber\\
&= D(1-P_{X|a_0z},P_{Y|b_0z}) +\!\!\!\!
\sum_{\genfrac{}{}{0pt}{}{a\in \cA_N, \, b \in \cB_N}{|a-b|=1}}\!\!\!\! D(P_{X|az},P_{Y|bz})
\nonumber\\
&\geq D(1-P_{X|a_0z},P_{X|a_0z})\nonumber \\
&= 2D(P_{X|a_0b_0z}, P_{\bar{X}}) \label{ItoD} \, .
\end{align}
The first inequality follows from the fact that $D(P_{X|\Omega},
P_{Y|\Omega}) \leq P(X \neq Y | \Omega)$ for any event $\Omega$ (see
Lemma~\ref{lem:Dbound} in Appendix~\ref{sec_dist}).  Furthermore, we
have used the conditions $P_{X|abz}=P_{X|az}$ and
$P_{Y|abz}=P_{Y|bz}$, and the triangle inequality for $D$.  By
symmetry, this relation holds for all $a$ and $b$.  

We now take the average over $z$ on both sides of~\eqref{ItoD}. The
left-hand-side gives
\begin{align}
&\sum_{z}P_{Z|ab}(z)I_N(P_{XY|ABz})\nonumber\\
&=\sum_{z}P_{Z}(z)I_N(P_{XY|ABz})\nonumber\\
&=\sum_{z}P_{Z|a_0b_0}(z)P(X=Y|a_0,b_0,z)+\nonumber\\
&\sum_{\genfrac{}{}{0pt}{}{a\in \cA_N, \, b \in \cB_N}{|a-b|=1}}\sum_{z}P_{Z|ab}(z)P(X\neq
Y|a,b,z)\nonumber\\
&=P(X=Y|a_0,b_0)+\!\!\!\!\sum_{\genfrac{}{}{0pt}{}{a\in \cA_N, \, b \in \cB_N}{|a-b|=1}}\!\!\!\! P(X\neq
Y|a, b, c)\nonumber\\
&=I_N(P_{XY|AB}) \, ,
\end{align}
where we used the condition $P_{Z|ab}=P_{Z}$ several
times. Furthermore, taking the average on the right-hand-side
of~\eqref{ItoD} yields
\begin{align*}
  \sum_zP_{Z|ab}(z)D(P_{X|abz},P_{\bar{X}})=D(P_{XZ|ab},P_{\bar{X}}\times
P_{Z|ab}) \, ,
\end{align*}
which is equivalent to the left-hand side of~\eqref{eq:markov2}. 
\end{proof}

We now apply Lemma~\ref{lem:2} to the quantum correlations
$P^N_{XY|ab\psi_2}$ arising from measurements on the maximally
entangled state $\psi_2$ (c.f.\ Section~\ref{sec_CB}). In the limit
where $N$ tends to infinity, we have
$\lim_{N\rightarrow\infty}I_N(P^N_{XY|AB})=0$, and hence we can
establish that $P_{X|abz}=P^N_{X|ab\psi_2}$ for all $a$, $b$ and $z$
with $P^N_{ABZ|\psi_2}(a,b,z)>0$.  Under the freedom of choice
assumption and assuming compatibility with quantum theory (note that
$P^N_{X|ab\psi_2}(x)=P_{\bar{X}}(x)=\frac{1}{2}$ for both $x=0$ and
$x=1$) this implies $P_{X|az}=P^N_{X|a\psi_2}$ for all $a$ and $z$
with $P_{Z|a}(z)>0$.  This means that $Z$ gives no additional
information about the measurement outcome, $X$.

Taking Parts~I and~II together, we obtain the following lemma, which
may be of independent interest.

\begin{lemma}[No higher theories give improved predictions for
  pairs of maximally entangled qubits] \label{thm_entanglementpredict} 
  For any $\delta>0$ there exists an $N \in \mathbb{N}$ such that for
  any RVs $A$, $B$, $X$, $Y$ and $Z$, at least one of the following
  three conditions cannot hold:
\begin{itemize}
  \item \emph{Freedom of choice:} $A$ and $B$ are free with respect to
    any of the causal orders depicted in
    Figure~\ref{fig:chronology};
  \item \emph{Compatibility with quantum theory:} $P_{XY|ABZ}$ is
    compatible with $P_{XY|AB\psi_2}^N$;\footnote{Note that this
      condition is (by definition) only satisfied if $P_A$ and $P_B$
      have full support.}
  \item \emph{Improved predictions:} There exists a value $A = a$ such
    that $\langle D(P_{X|az},P_{X|a\psi_2})\rangle_z>\delta$,
    where $\langle\cdot\rangle_z$ denotes the expectation value over
    $z$.
  \end{itemize}
\end{lemma}

Hence, if an alternative theory is compatible with quantum theory and
satisfies the freedom of choice assumption then the third condition
cannot hold, i.e., $\langle D(P_{X|a
  z},P_{X|a\psi_2})\rangle_z\leq\delta$. Since $\delta$ can be
arbitrarily small, this implies that quantum theory is as informative
as the alternative theory.

\subsection{Part~III: Generalization to arbitrary measurements}

The last part of the proof of Theorem~\ref{thm_main} consists of
generalizing Lemma~\ref{thm_entanglementpredict}, which applies to
specific measurements on a maximally entangled state, to measurements
on the general state $\ket{\phi}_{S D}$.  The proof relies on the
concept of embezzling states~\cite{DamHay03}. These are entangled
states that can be used to extract any desired maximally entangled
state locally and without communication. More precisely, we will use
the following lemma, which is implicit in~\cite{DamHay03}.

\begin{lemma} \label{lem_embezzling} For any $\delta > 0$ and for any
  $k \in \mathbb{N}$ there exists a bipartite state
  $\ket{\Gamma^{k}}_{\tilde{S} \tilde{D}}$, the \emph{embezzling
    state}, such that for any $m \leq k$, there exist local
  isometries, $U_m$ and $V_m$, on $\tilde{S}$ and $\tilde{D}$,
  respectively, that perform the transformation
  \begin{align*}
    U_m\otimes V_m: \, \ket{\Gamma^{k}}_{\tilde{S}\tilde{D}}\mapsto\ket{\Gamma^{k}}_{\tilde{S}\tilde{D}}\otimes\ket{\psi_m}_{S'D'}
  \end{align*}
  with fidelity at least $1-\delta$, where $\ket{\psi_m}_{S'D'}:=\frac{1}{\sqrt{m}}\sum_{x=0}^{m-1}\ket{\hat{x}}_{S'}\ket{\hat{x}}_{D'}$ denotes a maximally
  entangled state of two $m$ dimensional systems.
\end{lemma}

Note that the state $\ket{\phi}_{S D}$ considered in
Theorem~\ref{thm_main} can be represented by its Schmidt decomposition
as
\begin{align*}
  \ket{\phi}_{S D} =  \sum_{\hat{y}}\sqrt{p_{\hat{y}}} \ket{\hat{y}}_S \ot
  \ket{\hat{y}}_D \ .
\end{align*}
We now consider an embezzling state on $\tilde{S}\tilde{D}$ and use
Lemma~\ref{lem_embezzling} to define isometries $\hat{U}$ and
$\hat{V}$ on $S \tilde{S}$ and $D \tilde{D}$, respectively, which are
controlled by the entry $\hat{y}$ in the registers $S$ or $D$, and
build up entanglement between registers $S'$ and $D'$, i.e.,
\begin{align*}
  \hat{U} & = \sum_{\hat{y}} 
  \proj{\hat{y}}_S \otimes U_{m(\hat{y})}\\
  \hat{V} & = \sum_{\hat{y}}  
  \proj{\hat{y}}_D \otimes V_{m(\hat{y})} \ .
\end{align*}
The integers $m(\hat{y})$ are chosen such that the state resulting
from applying $\hat{U} \otimes \hat{V}$ to ${\ket{\phi}_{S D}
  \otimes \ket{\Gamma^k}_{\tilde{S} \tilde{D}}}$ is close to a state
of the form
\begin{align*}
  \Bigl(2^{-{r/2}}\sum_{\hat{y}}\sum_{\hat{y}'=0}^{m(\hat{y})-1}\ket{\hat{y},\hat{y}'}_{SS'}\otimes\ket{\hat{y},\hat{y}'}_{DD'}\Bigr)\otimes\ket{\Gamma^k}_{\tilde{D}\tilde{S}}
    \, ,
\end{align*}
with $\sum_{\hat{y}}m(\hat{y})=2^r$, for some integer $r$. (This can
be achieved to arbitrary precision for sufficiently large $k$ and
$m(\hat{y})$.)  Note that the first part of this state corresponds to
$r$ maximally entangled pairs, $\ket{\psi_2}^{\otimes r}$, between the
registers $S S'$ and $D D'$.  \footnote{As a simple example, consider
  the state
  $\ket{\phi}_{SD}=\frac{1}{2}\ket{\hat{0}}_S\ket{\hat{0}}_D+\frac{\sqrt{3}}{2}\ket{\hat{1}}_S\ket{\hat{1}}_D$.
In this case we would take $m(0)=1$ and $m(1)=3$ to yield a state of
the form
$\frac{1}{2}(\ket{\hat{0}\hat{0}}_{SS'}\ket{\hat{0}\hat{0}}_{DD'}+\ket{\hat{1}\hat{0}}_{SS'}\ket{\hat{1}\hat{0}}_{DD'}+\ket{\hat{1}\hat{1}}_{SS'}\ket{\hat{1}\hat{1}}_{DD'}+\ket{\hat{1}\hat{2}}_{SS'}\ket{\hat{1}\hat{2}}_{DD'})$
after the transformation.}
We now construct the POVMs $\{E_x^a\}$
and $\{F_y^b\}$ by concatenating the operations $\hat{U}$ and
$\hat{V}$ with the projective measurements along the vectors
$\ket{(\frac{a_i}{2N}+x_i)\pi}$ and $\ket{(\frac{b_i}{2N} + y_i)\pi}$
introduced in Section~\ref{sec_CB}. More precisely, we define
\begin{align*}
  E_x^a & := \hat{U}^{\dagger}\cdot \Bigl[\Bigl( \bigotimes_{i=1}^r \proj{(\tfrac{a_i}{2N}+x_i)\pi}
  \Bigr)_{D D'} \! \! \! \! \!   \otimes \id_{\tilde{D}} \Bigr] \cdot \hat{U}
  \\
  F_y^b & := \hat{V}^{\dagger}\cdot \Bigl[ \Bigl( \bigotimes_{i=1}^r \proj{(\tfrac{b_i}{2N} + y_i)\pi}
  \Bigr)_{S S'}  \! \! \! \! \!  \otimes \id_{\tilde{S}} \Bigr]  \cdot \hat{V}
\end{align*}
with $a = (a_1, \ldots, a_r) \in \cA_N^{\times r}$ and $b= (b_1,
\ldots, b_r) \in \cB_N^{\times r} $, for some large $N$. In addition
we define
\begin{align*}
  F_y^{b_0} = \proj{\hat{y}}_{S} \otimes \id_{\tilde{S}} \ .
\end{align*}

Assume now that the freedom of choice as well as the compatibility
with quantum theory assumption are satisfied. Furthermore, let $X=
(\hat{X}, \hat{X}')$ and $Y =\hat{Y}$ be the outcomes of the
measurements $A = a_0 := (0, \ldots, 0)$ and $B = b_0$, respectively.
By choosing the orientation of the vectors $\ket{\up}$ and $\ket{\dn}$
of Section~\ref{sec_CB} appropriately, we can arrange it so that
quantum theory predicts that the outcomes of the measurements of $a_0$
and $b_0$ are in agreement, in the sense that $\hat{X} = Y$ holds with
probability $1$. Hence, together with the no-signalling conditions
(c.f.\ Lemma~\ref{lem:1}) we find that
\begin{align*}
 P_{Y|b_0(\psi \otimes \Gamma)} & = P_{\hat{X}|a_0(\psi \otimes \Gamma)} \\
  P_{Y|b_0 Z} & = P_{\hat{X}|a_0 Z}
\end{align*}
Lemma~\ref{thm_entanglementpredict} implies that, $P_{X|a_0(\psi
  \otimes \Gamma)}$ must be as informative as $P_{X|a_0 Z}$. In
particular, the same relation holds for the marginals of these
distributions, i.e, $P_{\hat{X}|a_0(\psi \otimes \Gamma)} $ is as
informative as $P_{\hat{X}|a_0 Z}$. Combining this with the above
identities we find that $ P_{Y | b_0 \psi } = P_{Y | b_0(\psi \otimes
  \Gamma)} $ is as informative as $P_{Y|b_0 Z}$, thus concluding the
proof of Theorem~\ref{thm_main}.

\section{Alternative theories are equivalent to quantum
  theory} \label{sec_equivalence}

In this section, we discuss an implication of the non-extendibility
theorem (Theorem~\ref{thm_main}) to a long-standing debate on the
nature of the quantum mechanical wave function.  The debate centres
around whether it should be interpreted as a subjective quantity, for
example a state of knowledge about some underlying physical reality,
or whether it should instead be interpreted as objective
(real).\footnote{Note that in some subjective interpretations
  (e.g.~\cite{CFS}) there is no underlying physical reality|the wave
  function is simply a state of knowledge about future measurement
  outcomes and nothing more.}
The wave function could be considered subjective if there existed an
alternative theory, with predictions based on a parameter $Z$, that is
at least as informative as quantum theory, and in which two different
wave functions, say $\psi$ and $\psi'$, are compatible with the same
value of the parameter, say $Z=z$.  Formally, this would mean that
there exist $z$, $\psi$, and $\psi'\neq\psi$ such that
$P_{Z\Psi}(z,\psi)>0$ and $P_{Z\Psi}(z,\psi')>0$.  This is sometimes
called a $\psi$-epistemic view of the wave function and contrasts with
the $\psi$-ontic, or objective, view~\cite{HarSpek} (we refer
to~\cite{CFS,Spekkens,LeiSpek} for arguments in favour of the
$\psi$-epistemic view).  In the latter, the wave function is uniquely
determined by the parameters of any alternative theory that is at
least as informative as quantum theory, i.e., there exists a
(deterministic) function, $f$ such that $\Psi=f(Z)$.

Our result is based on the following simple lemma, which asserts that,
if an alternative theory is equally informative as quantum theory then
the wave function is indeed uniquely determined by the parameter of
the alternative theory.

\begin{lemma}\label{lem:wavefn} Suppose $\{E^a_x\}$ form a
  tomographically complete set of POVMs, and $A$, $X$, $\Psi$ and $Z$
  are RVs such that:
  \begin{itemize}
  \item $A$ is a free choice with respect to a causal order in which
    $A\nbef Z$ and $A\nbef\Psi$.
  \item $P_{X|AZ}$ is at least as informative as $P_{X|A\Psi}$, where
    $P_{X|a\psi}(x)= \tr(E^a_x\psi)$.
  \item $P_{X|A\Psi}$ is at least as informative as
  $P_{X|AZ}$.
\end{itemize}
Then there exists a function, $f$, such that $\Psi=f(Z)$.
\end{lemma}

\begin{proof}
  If $P_{X|AZ}$ is at least as informative as $P_{X|A\Psi}$, then
  there exists a distribution $\bar{P}_{XZ\Psi|A}$ such that 
\begin{align*}
    P_{X|a\psi} & = \sum_{z} \bar{P}_{X Z|a \psi}(\cdot, z) \ \
    \text{$\forall\ a,\psi$}\\
      P_{X|a z} & = \sum_{\psi} \bar{P}_{X \Psi|a z}(\cdot, \psi) \ \
      \text{$\forall\ a, z$}.
  \end{align*}
  (we drop the bar on $P$ in the following, and simply use
  $P_{XZ\Psi|A}$ to denote this distribution).  We
  have $$P_{X|az\psi}=P_{X|az}$$ for all $a,z,\psi$ that have a
  non-zero joint probability, i.e., $P_{AZ\Psi}(a,z,\psi)>0$.
  Likewise, if $P_{X|A\Psi}$ is at least as informative as $P_{X|AZ}$
  then $$P_{X|az\psi}=P_{X|a\psi}$$ holds under the same condition.
  Combining these expressions gives
\begin{align}  \label{eq:consist}
  P_{X|a\psi}=P_{X|az} \, .
\end{align}
If $A$ is a free choice, we have $P_{AZ\Psi}=P_A\times P_{Z\Psi}$,
hence~\eqref{eq:consist} holds provided that $P_{Z\Psi}(z,\psi)>0$ and
$P_A(a)>0$.

Let now $z$, $\psi$ and $\psi'$ be such that $P_{Z\Psi}(z,\psi)>0$ and
$P_{Z\Psi}(z,\psi')>0$.  From~\eqref{eq:consist}, this implies
$P_{X|a\psi} = P_{X|a\psi'}$ for all $a$ such that $P_A(a)>0$.  Since
the set of measurements with $P_A(a)>0$ is tomographically complete,
this can only be satisfied if $\psi=\psi'$.  It hence follows that
there exists a function $f$ such that $\Psi=f(Z)$.
\end{proof}

Combining Theorem~\ref{thm_main} with Lemma~\ref{lem:wavefn}, we can
establish the main result of this section, which we state informally
as follows.\smallskip

\begin{claim}~\cite{CR_wavefn} \label{claim:2} In any alternative
  theory that is at least as informative as quantum theory and
  compatible with free choice (with respect to the discussed
  causal orders), there is a one-to-one correspondence between the
  parameters of the alternative theory and the quantum state (up to a
  possible removable degeneracy\footnote{Any degeneracy is
    \emph{removable} in the sense that it has no operational effect,
    i.e., one can define another theory without the degeneracy (but
    otherwise identical) without affecting the predictive power.}  in
  the parameters of the alternative theory).
\end{claim}

To establish this, as before, we use $Z$ to denote the parameters of
the higher theory.  Theorem~\ref{thm_main} shows that under the
freedom of choice assumption, quantum theory is at least as
informative as any alternative theory.  We hence satisfy the
conditions of Lemma~\ref{lem:wavefn} so find $\Psi=f(Z)$, for some
function $f$.  Furthermore, since $Z$ cannot improve the predictions
for any $\Psi=\psi$, any $z$ in $f^{-1}(\psi)$ must give identical
predictions. Hence, if $f^{-1}(\psi)$ contains more than one element,
this corresponds to a removable degeneracy in the parameters of the
alternative theory.

\subsection*{Related work}
An interpretation of the wave function as a subjective state of
knowledge about some underlying theory has also been ruled out by
Pusey \emph{et al.}~\cite{PBR} via a different argument using
different assumptions which we now summarize.  They consider the
preparation of multiple quantum systems, with states $\Psi_i$, where
each system is associated with a particular parameter in the higher
theory, $Z_i$.  Pusey \emph{et al.}\ assume that the joint
distribution of these is product, i.e.
\begin{equation}\label{eq:prod}
P_{Z_1Z_2\ldots\Psi_1\Psi_2\ldots}=P_{Z_1\Psi_1}\times
P_{Z_2\Psi_2}\times\ldots\, .
\end{equation}
Starting from this assumption, they show that there cannot exist two
distinct states, $\psi$ and $\psi'$, such that for each $i$ there
exists a value of $Z_i=z_i$ satisfying $P_{Z_i\Psi_i}(z_i,\psi)>0$ and
$P_{Z_i\Psi_i}(z_i,\psi')>0$.

We note that the product nature of the joint distribution,
Eq.~\eqref{eq:prod}, is related to free choice of preparation.  In
particular, it implies
\begin{align*}
  P_{\Psi_1Z_2\ldots Z_N\Psi_2\ldots\Psi_N}=P_{\Psi_1}\times
  P_{Z_2\ldots Z_N\Psi_2\ldots\Psi_N}\, .
\end{align*}
If we take the causal order to be such that $\Psi_i\nbef\Psi_j$
and $\Psi_i\nbef Z_j$ for $j\neq i$ (as would be natural if we make
space-like separated preparations), then this is equivalent to saying that
$\Psi_1$ can be chosen freely.

It was subsequently noted~\cite{Hall_PBR} that the separability
assumption can be weakened, in essence to the assumption that there
exists a particular set of parameters in the higher theory that are
compatible with every product state composed of $\psi$ and $\psi'$,
i.e., there exist values of the parameters, $z_1,\ldots,z_N$, such
that
$$P_{Z_1\ldots Z_N\Psi_1\ldots\Psi_N}(z_1,\ldots,z_N,\psi^{(}\phantom{}'\phantom{}^{)},\ldots,\psi^{(}\phantom{}'\phantom{}^{)})>0\, ,$$
where each $\psi^{(}\phantom{}'\phantom{}^{)}$ is independently either
$\psi$ or $\psi'$ (so that the above represents $2^N$ conditions).  This
condition can be further weakened~\cite{SF} such that the parameters
of the alternative theory for multiple systems need not be made up
only of the individual parts, but could be replaced or supplemented
with global parameters (provided these are also compatible with all
the product state preparations).

An alternative argument against an interpretation of the quantum state
as a state of knowledge about an underlying reality can be found
in~\cite{Hardy_PBR}.

We remark that some models in which the wave function is subjective
have been developed for restricted scenarios.  For example, by
modifying an earlier model by Bell~\cite{Bell_ch1}, Lewis \emph{et
  al.}\ constructed a model for a single qubit in which the wave
function is subjective, and extended that model to arbitrary
dimensions~\cite{LJBR}.  These models are not in conflict with
Claim~\ref{claim:2} because they treat only single systems, and cannot
be extended to bipartite scenarios while allowing for free choice with
respect to one of the causal orders of Figure~\ref{fig:chronology}.

\section{Discussion} \label{sec_discussion}

The main statements described in this \chapter{} about the
completeness of quantum theory are based on two assumptions. One of
them is that quantum theory is correct, and is implicit in the
question of completeness. The other is that of free choice within a
natural causal order.  It is worth commenting on the existence of
alternative models that are not compatible with this assumption.

A prominent example is the de Broglie-Bohm model~\cite{deBroglie,Bohm}
which recreates quantum correlations, providing higher explanation in
the form of hidden particle positions.  These can be thought of as
parameters of a higher theory that would allow perfect predictions of
the outcomes.  However, introducing these parameters comes at a price:
it is incompatible with the freedom of choice assumption of our
theorems. In fact, for the bipartite setting discussed above, if $Z$
includes the particle positions of the de Broglie-Bohm model, we have
some non-local behaviour, so that $P_{X|abz}=P_{X|az}$, for instance,
does not hold.  Thus, given Lemma~\ref{lem:1}, it follows that $A$ and
$B$ cannot be free choices with respect to any of the causal orders of
Figure~\ref{fig:chronology}.

There are at least two ways to avoid our conclusions.  The first is to
maintain free choice, but assume that the alternative theory has a
different causal order (in particular, one in which either $A\nbef Y$ or
$B\nbef X$ does not hold).  The second is to consider alternative
theories without free choice, in which the measurement settings $A$
and $B$ may depend on the additional parameters $Z$ (sometimes, this
view is argued for by imagining that the additional parameters are
permanently hidden).

One may take the view that the freedom of choice assumption, which
demands complete independence between the chosen settings and the
other variables, is relatively strong, and perhaps contemplate
alternative theories where this assumption is weakened.  Some results
in this direction can be found in~\cite{CR_free}, where a theorem
similar to Lemma~\ref{thm_entanglementpredict} is established under
a relaxed free choice assumption, and provided there is no signalling
at the level of the underlying theory.

Finally, we note that the result presented here has a generic
application in quantum cryptography.  Standard security proofs for
schemes such as quantum key distribution~\cite{BB84,Ekert} are based
on the assumption (usually not stated explicitly) that quantum theory
is complete.  If this were not the case, it could be that a scheme is
proven secure within quantum theory, yet an adversary can break it by
exploiting information available in a higher theory. However, the
non-extendibility theorem, Theorem~\ref{thm_main}, implies that it is
sufficient to make only the weaker assumption that quantum theory is
correct, since this implies completeness.

\acknowledgments 

We are grateful to Gijs Leegwater and Robert Spekkens for helpful
comments on an earlier draft. This work was supported by the Swiss
National Science Foundation (via grants PP00P2-128455, 20CH21-138799
(CHIST-ERA project CQC) and 200020-135048, via the Swiss National
Centre of Competence in Research ``Quantum Science and Technology''
and via the CHIST-ERA project DIQIP), by the European Research Council
(via grant 258932), by the German Science Foundation (via grant CH
843/2-1) and by the Swiss State Secretariat for Education and Research
supporting COST action MP1006.

%\bibliographystyle{naturemag}
%\bibliography{../../bibtex}

\begin{thebibliography}{10}
\expandafter\ifx\csname url\endcsname\relax
  \def\url#1{\texttt{#1}}\fi
\expandafter\ifx\csname urlprefix\endcsname\relax\def\urlprefix{URL }\fi
\providecommand{\bibinfo}[2]{#2}
\providecommand{\eprint}[2][]{\url{#2}}

\bibitem{EPR}
\bibinfo{author}{Einstein, A.}, \bibinfo{author}{Podolsky, B.} \&
  \bibinfo{author}{Rosen, N.}
\newblock \bibinfo{title}{Can quantum-mechanical description of physical
  reality be considered complete?}
\newblock \emph{\bibinfo{journal}{Physical Review}}
  \textbf{\bibinfo{volume}{47}}, \bibinfo{pages}{777--780}
  (\bibinfo{year}{1935}).

\bibitem{KS}
\bibinfo{author}{Kochen, S.} \& \bibinfo{author}{Specker, E.~P.}
\newblock \bibinfo{title}{The problem of hidden variables in quantum
  mechanics}.
\newblock \emph{\bibinfo{journal}{Journal of Mathematics and Mechanics}}
  \textbf{\bibinfo{volume}{17}}, \bibinfo{pages}{59--87}
  (\bibinfo{year}{1967}).

\bibitem{Bell_KS}
\bibinfo{author}{Bell, J.~S.}
\newblock \bibinfo{title}{On the problem of hidden variables in quantum
  mechanics}.
\newblock In \emph{\bibinfo{booktitle}{Speakable and unspeakable in quantum
  mechanics}}, chap.~\bibinfo{chapter}{1} (\bibinfo{publisher}{Cambridge
  University Press}, \bibinfo{year}{1987}).

\bibitem{Bell}
\bibinfo{author}{Bell, J.~S.}
\newblock \bibinfo{title}{On the {E}instein-{P}odolsky-{R}osen paradox}.
\newblock In \emph{\bibinfo{booktitle}{Speakable and unspeakable in quantum
  mechanics}}, chap.~\bibinfo{chapter}{2} (\bibinfo{publisher}{Cambridge
  University Press}, \bibinfo{year}{1987}).

\bibitem{CR_ext}
\bibinfo{author}{Colbeck, R.} \& \bibinfo{author}{Renner, R.}
\newblock \bibinfo{title}{No extension of quantum theory can have improved
  predictive power}.
\newblock \emph{\bibinfo{journal}{Nature Communications}}
  \textbf{\bibinfo{volume}{2}}, \bibinfo{pages}{411} (\bibinfo{year}{2011}).

\bibitem{CR_wavefn}
\bibinfo{author}{Colbeck, R.} \& \bibinfo{author}{Renner, R.}
\newblock \bibinfo{title}{Is a system's wave function in one-to-one
  correspondence with its elements of reality?}
\newblock \emph{\bibinfo{journal}{Physical Review Letters}}
  \textbf{\bibinfo{volume}{108}}, \bibinfo{pages}{150402}
  (\bibinfo{year}{2012}).

\bibitem{Pearle}
\bibinfo{author}{Pearle, P.~M.}
\newblock \bibinfo{title}{Hidden-variable example based upon data rejection}.
\newblock \emph{\bibinfo{journal}{Physical Review D}}
  \textbf{\bibinfo{volume}{2}}, \bibinfo{pages}{1418--1425}
  (\bibinfo{year}{1970}).

\bibitem{BC}
\bibinfo{author}{Braunstein, S.~L.} \& \bibinfo{author}{Caves, C.~M.}
\newblock \bibinfo{title}{Wringing out better {B}ell inequalities}.
\newblock \emph{\bibinfo{journal}{Annals of Physics}}
  \textbf{\bibinfo{volume}{202}}, \bibinfo{pages}{22--56}
  (\bibinfo{year}{1990}).

\bibitem{CHSH}
\bibinfo{author}{{Clauser}, J.~F.}, \bibinfo{author}{{Horne}, M.~A.},
  \bibinfo{author}{{Shimony}, A.} \& \bibinfo{author}{{Holt}, R.~A.}
\newblock \bibinfo{title}{Proposed experiment to test local hidden-variable
  theories}.
\newblock \emph{\bibinfo{journal}{Physical Review Letters}}
  \textbf{\bibinfo{volume}{23}}, \bibinfo{pages}{880--884}
  (\bibinfo{year}{1969}).

\bibitem{Bell_free}
\bibinfo{author}{Bell, J.~S.}
\newblock \bibinfo{title}{Free variables and local causality}.
\newblock In \emph{\bibinfo{booktitle}{Speakable and unspeakable in quantum
  mechanics}}, chap.~\bibinfo{chapter}{12} (\bibinfo{publisher}{Cambridge
  University Press}, \bibinfo{year}{1987}).

\bibitem{Leggett}
\bibinfo{author}{Leggett, A.~J.}
\newblock \bibinfo{title}{Nonlocal hidden-variable theories and quantum
  mechanics: An incompatibility theorem}.
\newblock \emph{\bibinfo{journal}{Foundations of Physics}}
  \textbf{\bibinfo{volume}{33}}, \bibinfo{pages}{1469--1493}
  (\bibinfo{year}{2003}).

\bibitem{Bell_nouvelle}
\bibinfo{author}{Bell, J.~S.}
\newblock \bibinfo{title}{La nouvelle cuisine}.
\newblock In \emph{\bibinfo{booktitle}{Speakable and unspeakable in quantum
  mechanics}}, chap.~\bibinfo{chapter}{24} (\bibinfo{publisher}{Cambridge
  University Press}, \bibinfo{year}{2004}), \bibinfo{edition}{2nd} edn.

\bibitem{BBGKLLS}
\bibinfo{author}{Branciard, C.} \emph{et~al.}
\newblock \bibinfo{title}{Testing quantum correlations versus single-particle
  properties within {L}eggett's model and beyond}.
\newblock \emph{\bibinfo{journal}{Nature Physics}}
  \textbf{\bibinfo{volume}{4}}, \bibinfo{pages}{681--685}
  (\bibinfo{year}{2008}).

\bibitem{GPKBZAZ}
\bibinfo{author}{Gr\"oblacher, S.} \emph{et~al.}
\newblock \bibinfo{title}{An experimental test of non-local realism}.
\newblock \emph{\bibinfo{journal}{Nature}} \textbf{\bibinfo{volume}{446}},
  \bibinfo{pages}{871--875} (\bibinfo{year}{2007}).

\bibitem{ColbeckRenner}
\bibinfo{author}{Colbeck, R.} \& \bibinfo{author}{Renner, R.}
\newblock \bibinfo{title}{Hidden variable models for quantum theory cannot have
  any local part}.
\newblock \emph{\bibinfo{journal}{Physical Review Letters}}
  \textbf{\bibinfo{volume}{101}}, \bibinfo{pages}{050403}
  (\bibinfo{year}{2008}).

\bibitem{Hardy_ontbag}
\bibinfo{author}{Hardy, L.}
\newblock \bibinfo{title}{Quantum ontological excess baggage}.
\newblock \emph{\bibinfo{journal}{Studies In History and Philosophy of Modern
  Physics}} \textbf{\bibinfo{volume}{35}}, \bibinfo{pages}{267--276}
  (\bibinfo{year}{2004}).

\bibitem{Montina3}
\bibinfo{author}{Montina, A.}
\newblock \bibinfo{title}{Exponential complexity and ontological theories of
  quantum mechanics}.
\newblock \emph{\bibinfo{journal}{Physical Review A}}
  \textbf{\bibinfo{volume}{77}}, \bibinfo{pages}{022104}
  (\bibinfo{year}{2008}).

\bibitem{Montina}
\bibinfo{author}{Montina, A.}
\newblock \bibinfo{title}{State-space dimensionality in short-memory
  hidden-variable theories}.
\newblock \emph{\bibinfo{journal}{Physical Review A}}
  \textbf{\bibinfo{volume}{83}}, \bibinfo{pages}{032107}
  (\bibinfo{year}{2011}).

\bibitem{ChenMontina}
\bibinfo{author}{Chen, Z.} \& \bibinfo{author}{Montina, A.}
\newblock \bibinfo{title}{Measurement contextuality is implied by macroscopic
  realism}.
\newblock \emph{\bibinfo{journal}{Physical Review A}}
  \textbf{\bibinfo{volume}{83}}, \bibinfo{pages}{042110}
  (\bibinfo{year}{2011}).

\bibitem{Spekkens_context}
\bibinfo{author}{Spekkens, R.~W.}
\newblock \bibinfo{title}{Contextuality for preparations, transformations, and
  unsharp measurements}.
\newblock \emph{\bibinfo{journal}{Physical Review A}}
  \textbf{\bibinfo{volume}{71}}, \bibinfo{pages}{052108}
  (\bibinfo{year}{2005}).

\bibitem{BHK}
\bibinfo{author}{Barrett, J.}, \bibinfo{author}{Hardy, L.} \&
  \bibinfo{author}{Kent, A.}
\newblock \bibinfo{title}{No signalling and quantum key distribution}.
\newblock \emph{\bibinfo{journal}{Physical Review Letters}}
  \textbf{\bibinfo{volume}{95}}, \bibinfo{pages}{010503}
  (\bibinfo{year}{2005}).

\bibitem{BKP}
\bibinfo{author}{Barrett, J.}, \bibinfo{author}{Kent, A.} \&
  \bibinfo{author}{Pironio, S.}
\newblock \bibinfo{title}{Maximally non-local and monogamous quantum
  correlations}.
\newblock \emph{\bibinfo{journal}{Physical Review Letters}}
  \textbf{\bibinfo{volume}{97}}, \bibinfo{pages}{170409}
  (\bibinfo{year}{2006}).

\bibitem{DamHay03}
\bibinfo{author}{van Dam, W.} \& \bibinfo{author}{Hayden, P.}
\newblock \bibinfo{title}{Universal entanglement transformations without
  communication}.
\newblock \emph{\bibinfo{journal}{Physical Review A}}
  \textbf{\bibinfo{volume}{67}}, \bibinfo{pages}{060302(R)}
  (\bibinfo{year}{2003}).

\bibitem{CFS}
\bibinfo{author}{Caves, C.~M.}, \bibinfo{author}{Fuchs, C.~A.} \&
  \bibinfo{author}{Schack, R.}
\newblock \bibinfo{title}{Quantum probabilities as {B}ayesian probabilities}.
\newblock \emph{\bibinfo{journal}{Physical Review A}}
  \textbf{\bibinfo{volume}{65}}, \bibinfo{pages}{022305}
  (\bibinfo{year}{2002}).

\bibitem{HarSpek}
\bibinfo{author}{Harrigan, N.} \& \bibinfo{author}{Spekkens, R.~W.}
\newblock \bibinfo{title}{Einstein, incompleteness, and the epistemic view of
  quantum states}.
\newblock \emph{\bibinfo{journal}{Foundations of Physics}}
  \textbf{\bibinfo{volume}{40}}, \bibinfo{pages}{125--157}
  (\bibinfo{year}{2010}).

\bibitem{Spekkens}
\bibinfo{author}{Spekkens, R.~W.}
\newblock \bibinfo{title}{Evidence for the epistemic view of quantum states: A
  toy theory}.
\newblock \emph{\bibinfo{journal}{Physical Review A}}
  \textbf{\bibinfo{volume}{75}}, \bibinfo{pages}{032110}
  (\bibinfo{year}{2007}).

\bibitem{LeiSpek}
\bibinfo{author}{Leifer, M.~S.} \& \bibinfo{author}{Spekkens, R.~W.}
\newblock \bibinfo{title}{Formulating quantum theory as a causally neutral
  theory of {B}ayesian inference}.
\newblock \bibinfo{howpublished}{e-print \url{arXiv:1107.5849}}
  (\bibinfo{year}{2011}).

\bibitem{PBR}
\bibinfo{author}{Pusey, M.~F.}, \bibinfo{author}{Barrett, J.} \&
  \bibinfo{author}{Rudolph, T.}
\newblock \bibinfo{title}{On the reality of the quantum state}.
\newblock \emph{\bibinfo{journal}{Nature Physics}}
  \textbf{\bibinfo{volume}{8}}, \bibinfo{pages}{476--479}
  (\bibinfo{year}{2012}).

\bibitem{Hall_PBR}
\bibinfo{author}{Hall, M. J.~W.}
\newblock \bibinfo{title}{Generalisations of the recent
  {P}usey-{B}arrett-{R}udolph theorem for statistical models of quantum
  phenomena}.
\newblock \bibinfo{howpublished}{e-print \url{arXiv:1111.6304}}
  (\bibinfo{year}{2011}).

\bibitem{SF}
\bibinfo{author}{Schlosshauer, M.} \& \bibinfo{author}{Fine, A.}
\newblock \bibinfo{title}{On a recent quantum no-go theorem}.
\newblock \bibinfo{howpublished}{e-print \url{arXiv:1203.4779}}
  (\bibinfo{year}{2012}).

\bibitem{Hardy_PBR}
\bibinfo{author}{Hardy, L.}
\newblock \bibinfo{title}{Are quantum states real?}
\newblock \bibinfo{howpublished}{e-print \url{arXiv:1205.1439}}
  (\bibinfo{year}{2012}).

\bibitem{Bell_ch1}
\bibinfo{author}{Bell, J.~S.}
\newblock \bibinfo{title}{On the problem of hidden variables in quantum
  mechanics}.
\newblock In \emph{\bibinfo{booktitle}{Speakable and unspeakable in quantum
  mechanics}}, chap.~\bibinfo{chapter}{1} (\bibinfo{publisher}{Cambridge
  University Press}, \bibinfo{year}{1987}).

\bibitem{LJBR}
\bibinfo{author}{Lewis, P.~G.}, \bibinfo{author}{Jennings, D.},
  \bibinfo{author}{Barrett, J.} \& \bibinfo{author}{Rudolph, T.}
\newblock \bibinfo{title}{Distinct quantum states can be compatible with a
  single state of reality}.
\newblock \emph{\bibinfo{journal}{Physical Review Letters}}
  \textbf{\bibinfo{volume}{109}}, \bibinfo{pages}{150404}
  (\bibinfo{year}{2012}).

\bibitem{deBroglie}
\bibinfo{author}{de~Broglie, L.}
\newblock \bibinfo{title}{La m\'ecanique ondulatoire et la structure atomique
  de la mati\`ere et du rayonnement}.
\newblock \emph{\bibinfo{journal}{Journal de Physique, Serie VI}}
  \textbf{\bibinfo{volume}{VIII}}, \bibinfo{pages}{225--241}
  (\bibinfo{year}{1927}).

\bibitem{Bohm}
\bibinfo{author}{Bohm, D.}
\newblock \bibinfo{title}{A suggested interpretation of the quantum theory in
  terms of ``hidden'' variables. {I}}.
\newblock \emph{\bibinfo{journal}{Physical Review}}
  \textbf{\bibinfo{volume}{85}}, \bibinfo{pages}{166--179}
  (\bibinfo{year}{1952}).

\bibitem{CR_free}
\bibinfo{author}{Colbeck, R.} \& \bibinfo{author}{Renner, R.}
\newblock \bibinfo{title}{Free randomness can be amplified}.
\newblock \emph{\bibinfo{journal}{Nature Physics}}
  \textbf{\bibinfo{volume}{8}}, \bibinfo{pages}{450–--454}
  (\bibinfo{year}{2012}).

\bibitem{BB84}
\bibinfo{author}{Bennett, C.~H.} \& \bibinfo{author}{Brassard, G.}
\newblock \bibinfo{title}{Quantum cryptography: Public key distribution and
  coin tossing}.
\newblock In \emph{\bibinfo{booktitle}{Proceedings of IEEE International
  Conference on Computers, Systems, and Signal Processing}},
  \bibinfo{pages}{175--179}. \bibinfo{organization}{IEEE}
  (\bibinfo{publisher}{New York}, \bibinfo{year}{1984}).

\bibitem{Ekert}
\bibinfo{author}{Ekert, A.~K.}
\newblock \bibinfo{title}{Quantum cryptography based on {B}ell's theorem}.
\newblock \emph{\bibinfo{journal}{Physical Review Letters}}
  \textbf{\bibinfo{volume}{67}}, \bibinfo{pages}{661--663}
  (\bibinfo{year}{1991}).

\end{thebibliography}

%\newpage

\appendix

\section{Variational Distance} \label{sec_dist}

The following is a list of the main properties of the variational
distance $D(\cdot, \cdot)$  used in this work:
\begin{itemize}
\item $D(\cdot, \cdot)$ is a metric on the space of probability
  distributions.
\item $D(\cdot, \cdot)$ is upper bounded by $1$.
\item The variational distance of marginal distributions cannot be
  larger than that of the joint distributions: $D(P_X,Q_X)\leq
  D(P_{XY},Q_{XY})$ for any $P_{XY}$ and $Q_{XY}$.
\item It is convex: If $\{\alpha_i\}$ satisfy $\alpha_i\geq 0$ and
  $\sum_i\alpha_i=1$, and $\{P^i_X\}$ and $\{Q^i_X\}$ are sets of
  distributions over $X$, then $D(\sum_i\alpha_i P^i_X,
  \sum_i\alpha_i Q^i_X)\leq\sum_i \alpha_i D(P^i_X,Q^i_X)$.
\item For a joint distribution $P_{XY}$, the variational distribution
  of the marginal distributions is bounded by the probability that the
  RVs $X$ and $Y$ have different values: $D(P_X,P_Y)\leq P(X\neq Y)$.
\end{itemize}

The first four properties follow straightforwardly from the
definition.  The last is proved in the following.

\begin{lemma} \label{lem:Dbound} Let $X$ and $Y$ be two random
  variables jointly distributed according to $P_{XY}$. Then the
  variational distance between the marginal distributions $P_X$ and
  $P_Y$ is bounded by
  \begin{align*}
    D(P_X, P_Y) \leq P(X \neq Y) \, .
  \end{align*}
\end{lemma}

\begin{proof}
  Let $P_{XY}^{\neq}:=P_{XY|X\neq Y}$ be the joint distribution of
  $X$ and $Y$ conditioned on the event that they are not
  equal. Similarly, define $P_{XY}^= :=P_{XY|X=Y}$. We then
  have
  \begin{align*}
    P_{XY} = p_{\neq} P_{XY}^{\neq} + (1-p_{\neq}) P_{XY}^=
  \end{align*}
  where $p_{\neq} := P(X\neq Y)$.  By linearity, the marginals of these
  distributions satisfy the same relation, i.e.,
\begin{eqnarray*}
P_{X} &=& p_{\neq} P_{X}^{\neq} + (1-p_{\neq}) P_{X}^= \\ 
%\text{and}\ \ \  
    P_{Y}&=&p_{\neq} P_{Y}^{\neq} + (1-p_{\neq}) P_{Y}^=\, .
\end{eqnarray*}
  Hence, by convexity of the variational distance,
  \begin{eqnarray*}
    D(P_X, P_Y) 
  &\leq& p_{\neq} D(P_X^{\neq}, P_Y^{\neq}) + (1-p_{\neq})  D(P_X^=,
  P_Y^=) \\
  &\leq& p_{\neq} \, ,
  \end{eqnarray*}  
  where the last inequality follows because the variational distance
  is at most 1, and $D(P_X^=, P_Y^=) = 0$.
\end{proof}

\end{document}